\documentclass[10pt,journal,compsoc]{IEEEtran}

\ifCLASSINFOpdf

\else

\fi

\ifCLASSOPTIONcompsoc
  \usepackage[nocompress]{cite}
\else
  \usepackage{cite}
\fi
\usepackage{bm}

\usepackage{amsmath}
\usepackage{epsf}
\usepackage{graphics}
\usepackage{ amssymb }
\usepackage[dvips]{graphicx}
\usepackage{epsfig}
\usepackage{cite}
\usepackage[linesnumbered,ruled,vlined]{algorithm2e}
\usepackage{colortbl}
\usepackage{color}
\usepackage{dblfloatfix}
\usepackage{bm}
\usepackage{amsmath}
\usepackage{epsf}
\usepackage{graphics}
\usepackage{ amssymb }
\usepackage[dvips]{graphicx}
\usepackage{epsfig}
\usepackage{cite}
\usepackage{graphicx}
\usepackage{epsfig}
\usepackage{latexsym}
\usepackage{amsfonts}
\usepackage{here}
\usepackage{rawfonts}

\usepackage[utf8]{inputenc}
\usepackage[english]{babel}
\usepackage{amsmath}
\usepackage{amsfonts}
\usepackage{amssymb}
\usepackage{color} 
\usepackage{bm}
\usepackage{listings}
\usepackage{amssymb}
\usepackage{amsthm}
\usepackage{graphicx}
\usepackage{epstopdf}
\usepackage{listings}
\usepackage{float}
\usepackage{amsmath}
\usepackage{amssymb}
\usepackage{amsfonts}
\usepackage{epstopdf}
\usepackage[overload]{empheq}

\usepackage{multirow}
\usepackage{amscd}
\usepackage{mathrsfs}
\usepackage{graphicx}
\usepackage{color}
\usepackage{url}
\usepackage{bm}
\usepackage{setspace}

\usepackage{footnote}
\DeclareMathOperator*{\argmax}{arg\,max}

\DeclareMathOperator*{\argmin}{arg\,min}
\newcommand*{\Scale}[2][4]{\scalebox{#1}{$#2$}}%
\newtheorem{theorem}{Theorem}
\newtheorem{lemma}{Lemma}

\newtheorem{proposition}{Proposition}

\newtheorem{remark}{Remark}
\newtheorem{fact}{Fact}

\usepackage[]{algorithm2e}

\usepackage{lipsum,graphicx,subcaption}
\usepackage[protrusion=true,expansion=true]{microtype}
\usepackage[font=small]{caption}
\ifCLASSINFOpdf

\else

\fi

\addto\captionsenglish{}
\usepackage[linesnumbered]{algorithm2e}
\makeatletter
\newcommand{\nosemic}{\renewcommand{\@endalgocfline}{\relax}}
\newcommand{\dosemic}{\renewcommand{\@endalgocfline}{\algocf@endline}}
\let\oldnl\nl
\newcommand{\nonl}{\renewcommand{\nl}{\let\nl\oldnl}}
\makeatother

    \makeatletter
\newcommand{\removelatexerror}{\let\@latex@error\@gobble}
\makeatother
\begin{document}

	\title{Prevention and Mitigation of Catastrophic Failures in Demand-Supply Interdependent Networks}
		\author{Seyyedali Hosseinalipour,~\IEEEmembership{Student Member,~IEEE,}
        Jiayu Mao,
        Do Young Eun,~\IEEEmembership{Senior Member,~IEEE,}
        and Huaiyu Dai,~\IEEEmembership{Fellow,~IEEE}
\IEEEcompsocitemizethanks{\IEEEcompsocthanksitem S. Hosseinalipour, D. Y. Eun and H. Dai are with the Department
of Electrical and Computer Engineering, North Carolina State University, Raleigh,
NC 27695 (Emails: shossei3,dyeun,hdai@ncsu.edu).\protect\\
\IEEEcompsocthanksitem Jiayu Mao is with the Department of Electrical Engineering, Penn State University, State College, PA 16801 (Email: jbm6279@psu.edu).}}
	\IEEEtitleabstractindextext{%
		\begin{abstract}
We propose a generic system model for a special category of interdependent networks, demand-supply networks, in which the demand and the supply nodes are associated with heterogeneous loads and resources, respectively. Our model sheds a light on a unique cascading failure mechanism induced by resource/load fluctuations, which in turn opens the door to conducting \textit{stress analysis} on interdependent networks. Compared to the existing literature mainly concerned with the node connectivity, we focus on developing effective resource allocation methods to prevent these cascading failures from happening and to mitigate/confine them upon occurrence in the network. To prevent cascading failures, we identify some dangerous stress mechanisms, based on which we quantify the \textit{robustness} of the network in terms of the resource configuration scheme. Afterward, we identify the optimal resource configuration under two resource/load fluctuations scenarios: uniform and proportional fluctuations. We further investigate the optimal resource configuration problem considering heterogeneous resource sharing costs among the nodes. To mitigate/confine ongoing cascading failures, we propose two network adaptations mechanisms: \textit{intentional failure} and \textit{resource re-adjustment}, based on which we propose an algorithm to mitigate an ongoing cascading failure while reinforcing the surviving network with a high robustness to avoid further failures.
		\end{abstract}
		
		}

	\maketitle

	\IEEEdisplaynontitleabstractindextext

	\IEEEpeerreviewmaketitle

%
%

	\IEEEraisesectionheading{\section{Introduction}}
\noindent	\IEEEPARstart{I}{nterdependent} networks are a relatively new concept in networking studies, which consist of two or more networks/layers with dependent functionality. Some well-known examples of these networks are i) smart power grid networks, where the functionality of the power-related sensors and the communication infrastructure are coupled \cite{CatastrophicCascadeofFailuresInInterdependentNetworks}, and ii) complex transportation systems consisting of multiple coupled transportation layers, e.g., railways  and  roadways~\cite{TheRobustnessofInterdependentTransportationNetworksUnderTargetedAttack}. The main motivation to study robustness, failure mechanisms, and fault propagation (we refer to all of them as \textit{robustness analysis} for brevity) for these networks is their special failure mechanism called cascading failures~\cite{CatastrophicCascadeofFailuresInInterdependentNetworks}.

Demand-supply networks are a category of interdependent networks, in which the functionality/operability of the nodes in one layer (demand layer) is contingent on receiving supplies from those in the other layer (supply layer). 
Due to the freshness of the topic, there is limited literature on conducting robustness analysis in the context of demand-supply networks. The most related works are~\cite{CascadingFailuresinInterdependentNetworkswithMultipleSupply-DemandLinksandFunctionalityThresholds,arxiveConnectivity}.  In~\cite{CascadingFailuresinInterdependentNetworkswithMultipleSupply-DemandLinksandFunctionalityThresholds}, a model for demand-supply networks is proposed, where the functionality of the demand nodes depends on the number of connected supply nodes. Their model can be considered as an abstract representation of some real-world economic systems such as the network of financial firms and non-financial companies. In the proposed model, the operability of a demand node is examined by solely counting the number of supply nodes connected to it. Hence, there is no notion of resource as a shared quantity among the nodes. In~\cite{arxiveConnectivity}, another model for demand-supply networks is proposed, where having at least a connection to the supply layer guarantees the operability of a demand node. Their model has some similarity to the cyber-physical networks and shared risk groups previously studied in~\cite{Inference,Diverse,approximability}.
In conclusion, in these works  the functionality of a demand node is examined solely by counting the number of supply nodes supporting it, and the notion of \textit{resource} as a quantitative value is overlooked. As a result, these works are mainly concerned with the connectivity among the nodes, while the underlying resource sharing/configuration mechanisms are ignored.

In this work, we conduct robustness analysis on a more generic model, in which the shared resource  among the connected nodes is represented by a quantity that varies from one pair of nodes to another. In our model, demand nodes and supply nodes are associated with heterogeneous \textit{requested loads} and heterogeneous \textit{resource provisioning} capabilities, respectively. In this paradigm, the functionality of a demand node is examined by the total amount of resource received from those connected supply nodes. This leads to the existence of (infinitely) many choices for the configuration of shared resources, i.e., \textit{resource configuration}, among the nodes for a network with a fixed node connectivity, all of which satisfying the requested loads of the demand nodes. This raises the following questions: i) is there any advantage in using a specific resource configuration? and if yes, ii) how important is the study of the resource configuration scheme to conduct robustness analysis? We demonstrate that the resource configuration scheme plays a key role in robustness analysis of our proposed model for demand-supply networks, which goes beyond the connectivity among nodes. One of the main advantages of our model is its adaptability to many real-world scenarios, e.g., allocation of infantry/capitals (considered as resources) to multiple divisions (considered as demand nodes with different requested loads) in battlefields, food/resources and humans in natural disasters, datacenters and users in a cloud network, banks and assets, and power plants and cities. 

  This paper can be broken down into two main parts: i) prevention of cascading failures, and ii) mitigation of ongoing cascading failures, in each of which we focus on studying the resource configuration schemes from a different angle. In the first part, considering quantitative values for the loads and resources, we analyze interdependent networks via \emph{stress analysis}, which cannot be readily carried out using the system model of~\cite{CascadingFailuresinInterdependentNetworkswithMultipleSupply-DemandLinksandFunctionalityThresholds,arxiveConnectivity}. We investigate the effect of \textit{fluctuations} incurred on both the requested loads and the resource provisioning capabilities on the stability of the network. We then propose two \textit{metrics} for the robustness of the network with regard to the resource configuration scheme. In this context, we aim to reinforce the network with a high robustness to prevent cascading failures from happening. In the second part, we target studying the distinct problem of stopping/confining ongoing cascading failures by solely re-adjusting the utilized resource configuration during the propagation of the failures. In particular, we extend the recently introduced concept of \textit{network adaptability} in interdependent networks, e.g.,~\cite{RecoveryOfInterdependentNetworks,StrategyForStoppingFailureCascadesInInterdependentNetworks}, to our model for demand-supply networks and propose effective mechanisms for adaptation of the resource configuration to stop cascading failures. For better comprehension, more discussion on the related works is postponed to Section~\ref{relatedwork}.
  
Our contributions can be summarized as follows: 
\vspace{-2mm}
\begin{enumerate}
\item We provide a generic system model for demand-supply networks capturing heterogeneous resource provisioning capabilities of the supply nodes and heterogeneous requested loads of the demand nodes.  In this paradigm, we reveal the effect of the resource sharing/configuration protocol on the reliability of the network and identify some dangerous stress mechanisms, which can lead to catastrophic failures.

\item As compared to the conventional cascading failures, we study an extended cascading failure process triggered by resource/load fluctuations, taking into account the overload of the supply nodes and resource deficiency of the demand nodes. 

\item  We quantify the robustness of the network with regard to its tolerability against fluctuations in both resource provisioning capabilities and requested loads under two resource/load fluctuation scenarios: i)~uniform fluctuations and ii) proportional fluctuations.  Afterward, we draw a connection between the prevention of cascading failures and maximization of the introduced robustness metrics. 

\item We determine the optimal resource configuration protocol to maximize the robustness for each of the two aforementioned scenarios of resource/load fluctuations. We also address the cost-effective design achieving the highest robustness considering heterogeneous geographical constraints/limitations on resource sharing among the nodes reflected through heterogeneous resource sharing costs. Moreover, assuming a budget constraint, we propose an algorithm to decrease the cost of resource provisioning while maintaining a high robustness for the network. 

\item We introduce a new perspective to the concept of \textit{network adaptability} with respect to the underlying resource configuration protocol in demand-supply networks. In this regard, we propose two adaptability mechanisms: \textit{intentional failure} and \textit{resource re-adjustment}, based on which we develop a novel algorithm to confine/mitigate ongoing cascading failures and provide high robustness for the surviving network so as to avoid triggering of further failures.    
\end{enumerate}
   	\begin{figure*}[t]
		\includegraphics[width=6.0in,height=1.40in]{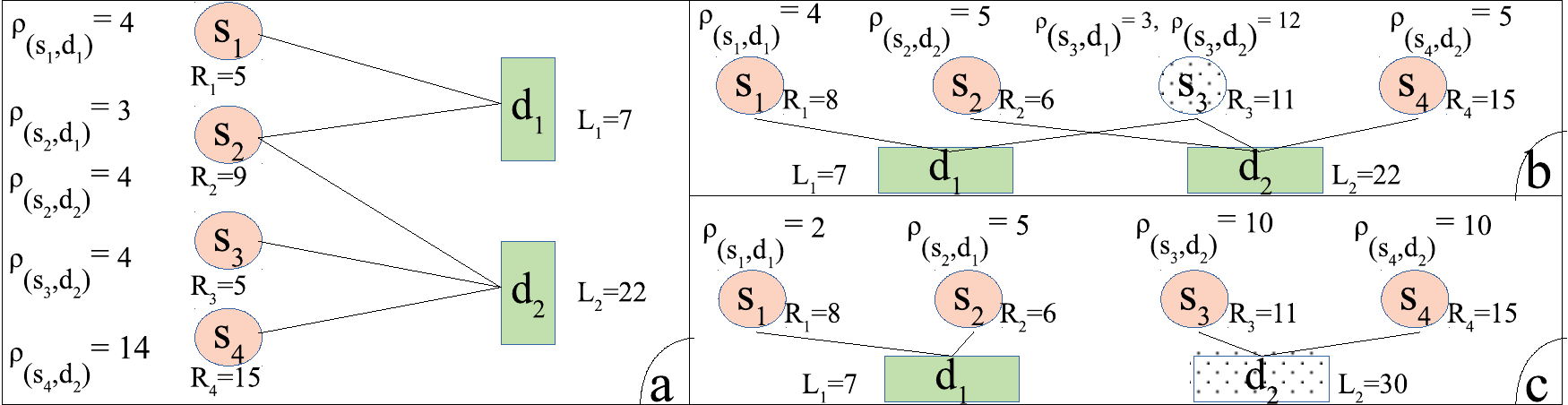}
		\centering
		\caption{a: An example of a stable network. b: An example of a network with an overloaded supply node, c: An example of a network with resource deficiency in a demand node.}
		\label{diag:show1}
	\end{figure*}
    	\begin{figure*}[!b]
		\includegraphics[width=6.0in,height=1.45in]{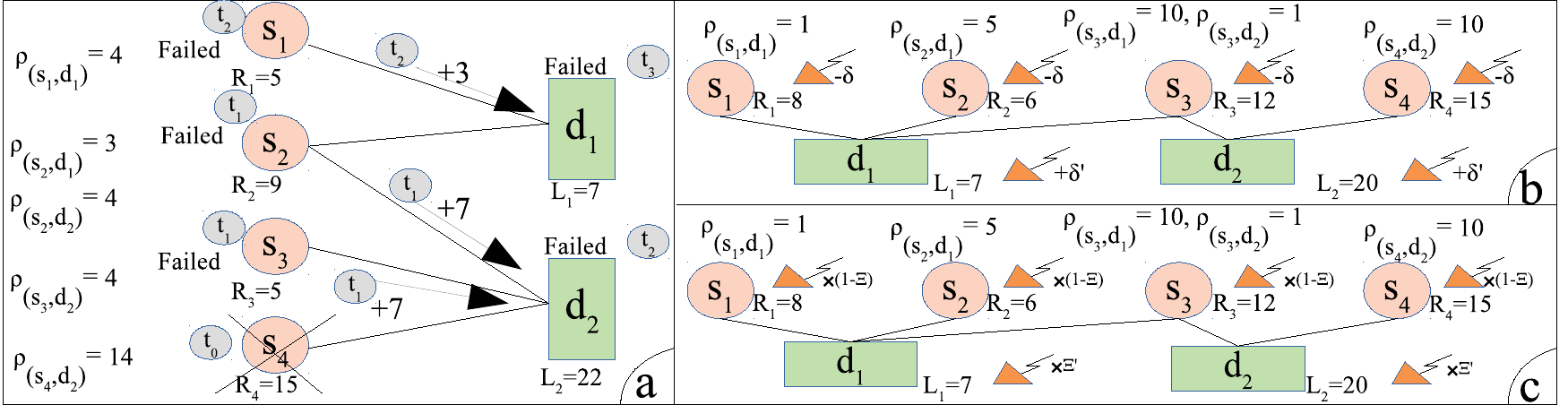}
		\centering
		\caption{a: An example of cascading failures initiated from internal failure of supply node $s_4$ (the numbers in the gray circles indicate time indices). b: Uniform resource/load fluctuations. c: Proportional resource/load fluctuations.}
		\label{diag:show2}
	\end{figure*}
	
	\vspace{-3mm}
\section{Preliminaries}\label{sec:Sysmodel}
\subsection{System model}\label{sec:Sysmodelsys}
We propose a model in which a demand-supply network comprises two layers: supply and demand layer. In the supply layer, there is a set of supply nodes denoted by $\mathcal{S}=\{s_1,\cdots,s_S\}$, $|\mathcal{S}|=S$, where each supply node $s_k\in \mathcal{S}$ is associated with some \textit{resource} units $R_k\in\mathbb{R}^+$. In the demand layer, a set of demand nodes $\mathcal{D}=\{d_1,\cdots,d_D\}$, $|\mathcal{D}|=D$, is considered, where each demand node $d_k \in \mathcal{D}$ is associated with a requested \textit{load} $L_k\in\mathbb{R}^+$. In our model, a supply node can distribute its resources among multiple demand nodes, each of which is capable of receiving resources from multiple supply nodes.\footnote{The geographical restrictions affecting the resource provisioning protocol are investigated in Section~\ref{sec:cost-based}.} For demand node $d_i\in\mathcal{D}$, let $\rho_{(s_k,d_i)}$ denote the amount of resources received from supply node $s_k \in \mathcal{S}$. We represent a demand-supply network as a weighted bipartite graph in which each supply node $s_k\in\mathcal{S}$ is connected to a demand node $d_i\in\mathcal{D}$ if $\rho_{(s_k,d_i)}>0$.\footnote{We will define the weights of the edges with respect to different mechanisms of resource/load fluctuations in the following section.} 

We study the network behavior under the following \textit{stability} conditions:
\begin{align}
&r_k \triangleq \sum_{d_i\in\mathcal{D}} \rho_{(s_k,d_i)}\leq R_k & \forall s_k \in \mathcal{S},\label{stableSupply}\\
&l_i \triangleq\sum_{s_k\in \mathcal{S}}\rho_{(s_k,d_i)}\geq L_i & \forall d_i \in \mathcal{D},\label{stabledemand}
\end{align}
where $r_k$ denotes the \textit{aggregate offered resources} of supply node $s_k\in\mathcal{S}$, and $l_i$ indicates the \textit{aggregate received resources} of demand node $d_i\in\mathcal{D}$. In our model, violations in~\eqref{stableSupply} are regarded as \textit{overloading} of the supply node(s), whereas violations in~\eqref{stabledemand} are considered as \textit{resource deficiencies} in the demand node(s). Since the system will always suffer from instability if the total available resources becomes less than the total requested loads, to avoid triviality, we assume $\sum_{k=1}^{S} R_k > \sum_{k=1}^{D} L_k$ throughout the paper. For a small network consisting of $6$ nodes, three examples of resource sharing/configuration among the nodes are depicted in Fig.~\ref{diag:show1}. It can be verified that the stability conditions hold only for the left instance (Fig.~\ref{diag:show1} (a)).
	
We conduct \textit{stress analysis} on demand-supply networks by classifying the potential stress mechanisms into the following categories: i) internal failure of demand nodes; ii) internal failure of supply nodes; iii) load fluctuations (increase/decrease) at the demand nodes; and iv) resource fluctuations (increase/decrease) at the supply nodes. In the following, harmless stresses are discussed first followed by the dangerous stresses. Considering the stability conditions (\eqref{stableSupply},~\eqref{stabledemand}), the first stress mechanism is harmless since it neither overloads the resource nodes (\eqref{stableSupply} remains intact) nor affects the receiving loads of the other demand nodes (\eqref{stabledemand} keeps inviolate). Also, a reduction in the load of a demand node is harmless since it has no effect on the stability of the network. A similar reason makes increment in the resources of supply nodes harmless. 

In contrast, internal failure of a supply node $s_k\in \mathcal{S}$ is potentially a dangerous stress causing $\rho_{(s_k,d_i)}=0, \forall d_i \in \mathcal{D}$, which may result in resource deficiencies in demand nodes (violation of~\eqref{stabledemand}). Another dangerous stress mechanism is a reduction in resource of a supply node, which (potentially) promotes overloading of the supply node (violation of~\eqref{stableSupply}) and, upon reduction in the amount of offered resources to stabilize the node, resource deficiencies in those connected demand nodes (violation of~\eqref{stabledemand}). Similarly, an increase in the load of a demand node is dangerous since it may (potentially) lead to resource deficiency in the node (violation of~\eqref{stabledemand}), and subsequently overloading of those connected supply nodes (violation of~\eqref{stableSupply}). These stress mechanisms promote a unique cascading failure process, which spreads differently as compared to the conventional cascading failures \cite{CatastrophicCascadeofFailuresInInterdependentNetworks},  discussed in the following subsection. 
\subsection{Cascading failures triggered by dangerous stress mechanisms}\label{sec:SysmodelsysCas}
An example of initialization and evolution of a cascading failure caused by instability of the nodes is depicted in the left diagram of Fig.~\ref{diag:show2} (Fig.~\ref{diag:show2} (a)), where the time indices written in gray circles indicate the order of actions ($t_0<t_1<t_2<t_3$). The cascading failure is initially (at $t_0$) triggered with internal failure of the supply node $s_4$ resulting in resource deficiency in the demand node $d_2$. To trace the evolution of the cascading failure, we check the stability conditions (\eqref{stableSupply},~\eqref{stabledemand}) at each time instance. To compensate for the resource deficiency, demand node $d_2$ drains $14$ extra unit of resources from its adjacent/connected supply nodes. Assume that $d_2$ drains $7$ extra unit of resources from both $s_2$ and $s_3$ (more elaborations on this subject is given in Section~\ref{sec:fluctuations}). At $t_1$, since $s_2$ and $s_3$ can only provide $2$ and $1$ extra unit of resources while being stable, they both fail resulting in resource deficiency in $d_1$ (violation of~\eqref{stabledemand}). At $t_2$, $d_1$ drains $3$ extra unit of resources from $s_1$ leading to the failure of $s_1$ due to overloading (violation of~\eqref{stableSupply}). At the same time, $d_2$ fails due to failure of all of its connected supply nodes (violation of~\eqref{stabledemand}). At the next time instant ($t_3$), $d_1$ fails due to the same reason. In this example, the cascading failures led to failure of all the nodes; however, the influence could be different for another resource configuration scheme. 

We consider two resource/load fluctuation mechanisms triggering cascading failures: uniform and proportional fluctuations, which are shown in diagram (b) and (c) of Fig.~\ref{diag:show2}, respectively. In the uniform fluctuation scenario (diagram (b)), a portion of supply (demand) nodes experience the same amount of reduction (increment), denoted by $\delta$ ($\delta'$) in the diagram, in their resources (loads). In this scenario, the fluctuation in the load of a demand node is distributed evenly among the connected supply nodes. This happens specially when there is an agreement among the supply nodes to evenly compensate for resource deficiencies on the demand layer. Two example of this scenario are i) equal compensation for infantry/capitals among divisions in the battlefield context using multiple backup supply sources, and ii) compensation for the overload of the running cloud servers by providing equal extra processing power from backup resources. It can be verified that for a $\delta$ slightly greater than $1$ both $s_2$ and $s_3$ become unstable and fail. Also, assuming $\delta'$ to be $2$ for $d_2$ puts $s_3$ at the threshold of failure (see Section~\ref{hom}). In the proportional fluctuation scenario (diagram (c)), the amount of fluctuations of the resources/loads is proportional to their initial values, where nodes with larger resources (loads) experience larger reduction (increment), denoted by $\Xi$ ($\Xi'$) in the diagram. In this scenario, the fluctuation in the load of a demand node is distributed among the connected supply nodes proportional to their initial offered resources. Two examples of this scenario are i) banks-assets network, in which the fluctuation in values of the assets affects the banks with respect to their amount of investment, and ii) power plants and cities, in which the fluctuation in the load of the cities distributes among multiple power plants according to their supplying powers. It can be seen that considering $\Xi$ slightly larger than $1/12$ leads to failure of $s_3$. Also, $s_3$ fails when $\Xi'$ is slightly larger than $12/11$ for both $d_2$ and $d_1$ (see Section~\ref{het}). Note that, considering the mentioned dangerous stress mechanism, internal failures of supply nodes can be regarded as load fluctuations in the demand nodes, and thus omitted from our discussions. 

	  \subsection{Related work}\label{relatedwork}
 \subsubsection{Connectivity-based cascading failures}\label{subsub1}  \textit{Robustness} of interdependent networks is studied extensively in literature in the context of a special failure mechanism called \textit{cascading failures}~\cite{CatastrophicCascadeofFailuresInInterdependentNetworks}, most of which aim at designing robust interdependent networks against node/edge removals and studying the robustness with respect to the network characteristics, e.g.,~\cite{NetworkRobustnessofMultiplexNetworkswithInterlayerDegreeCorrelations,DesigningOptimalInterlinkPatternstoMaximizeRobustnessofInterdependentNetwork,ImprovingRobustnessofInterdependentNetworksbyaNewCouplingStrategy,CyberPhysical,markovChainCascades,huang2011robustness}. There is a parallel trend of research in D2D networks pursuing the same goal, e.g., \cite{ResilienceD2D1,ResilienceD2D2}. In these works, the functionality of a node is merely examined based on its connectivity to the giant component in its layer and its dependent nodes on the other layer.  {\color{black}Moreover, there are some similar works in the power-grid literature, e.g.,~ \cite{tootaghaj1,tootaghaj2}, which consider the cascade of failures between the power grid network and the communication network. In this literature, controlling the loads and the power generators in the power grid network is studied to suppress the impact of cascading failure. Nevertheless, the resource sharing/provisioning among the nodes between different layers is not studied in all the aforementioned works and most of contemporary literature.}

\subsubsection{Load-based cascading failures} Cascading failures caused by load (re-)distribution among the nodes are studied in literature (e.g.,~\cite{RobustnessofInterdependentNetworksWithDifferentLinkPatternsAgainstCascadingFailures,ERfailure,impactOfTopology}), in which the load of the nodes are independent/decoupled to each other and load distribution only occurs upon failures of the nodes. However, in our model, there is a constant coupling between resources provided from the supply nodes and the load of the demand nodes. Hence, fluctuations of the resources/loads in one layer can lead to instability of the nodes in the other layer. Moreover, in our model, in the middle of a cascading failure, a demand node with resource deficiency drains extra resources from its adjacent supply nodes, which can result in either stability of the demand node or overloading of the supply nodes. This paradigm distinguishes our cascading failure mechanism and proposed robustness analysis from that literature. 

\subsubsection{Network adaptability} The authors in \cite{NetworkAdaptabilityFromDisasterDisruptionsAndCascadingFailures} provide some general ideas and approaches to prepare the telecommunication networks priori to natural disasters. In the single layer network context, \cite{ImproveNetworksRobustnessAgainstCascadewithRewiring} proposes network rewiring mechanisms aiming to achieve a higher robustness against cascading failures. The concept of network adaptability is recently investigated in the context of interdependent networks in~\cite{RecoveryOfInterdependentNetworks,StrategyForStoppingFailureCascadesInInterdependentNetworks} in terms of rehabilitating a portion of failed nodes in the process of cascading failures. However, the system model and cascading failures mechanism considered in these works belong to the first enumerated category (Section~\ref{subsub1}) rendering them irrelevant to our work. In this paper, we provide a new perspective to the concept of network adaptability by incorporating the resource sharing mechanism into cascading failures. Our focus is on adaptation of resource allocations among the nodes to confine cascading failures. To this end, we propose two adaptation strategies: i) intentional failures of the demand nodes, and ii) resource re-adjustments (re-allocations) among the supply nodes, using which we propose an algorithm to confine cascading failures and provide high robustness for the surviving network.  
    
\subsubsection{Demand-supply networks}   There are a few works studying interdependent networks in the context of demand-supply networks, among which the most related ones are~\cite{CascadingFailuresinInterdependentNetworkswithMultipleSupply-DemandLinksandFunctionalityThresholds,arxiveConnectivity}. In these works, there is no notion of ``resource'' as a quantitative value and the demands of the demand nodes are merely identified by the required number of connected supply nodes. Consequently, studying the connectivity of the network is the main focus of these studies, e.g., extending the concept of the edge-cut problem~(\cite{ApproximationAndHardnessResultsForLabelCutandRelatedProblems,ComplexityAndApproximabilityIssuesOfSharedRiskResourceGroup}) to the node-cut problem in~\cite{arxiveConnectivity}. Also, the model in these works captures one-way dependency between the nodes, where the functionality of the supply nodes is presupposed irregardless of the demand layer situations. In this paper, considering resources/loads as coupled quantitative values imposes a two-way dependency among the nodes, where resources/loads fluctuations in one layer can lead to instability in the other. Also, in our model, two networks with the same connectivity (wiring) configuration may react to failures differently due to different resources/loads of the nodes. Due to these inherent differences, the cascading failure process of this work along with its proposed design/analysis is explicitly different from those works.   

We would like to mention notable works of~\cite{CascadingFailuresInbi-partiteGraphsModelForSystemicRiskPropagation,StabilityAnalysisofFinancialContagionDuetoOverlappingPortfolios} in economics literature, which propose bipartite interdependent network modeling of bank-asset networks. We contribute to this literature with our design framework achieving a high robustness against fluctuations and our network adaptation methods to confine cascading failures, which are complementary and can be integrated into their models.

{\color{black} Also, there is a body of works on modeling the load and capacity of the transmission lines in the power grid literature, the most relevant of which aim to model the cascading failures caused by load fluctuation on the transmission lines and propose effective network designs, e.g.,
\cite{UniformLines1,UniformLines2,UniformLines3}. In this literature, the network consists of multiple transmission lines; upon failure, the load of the failed lines will be redistributed among all the remaining ones. Hence, there is no notion of node (supply node/demand node) as an entity in the network. Consequently, the network is not considered as a two-layer interdependent network with resource provisioning among the nodes, making their model fundamentally different as compared to our bipartite network model. This fact makes our system model and cascading failure mechanism completely different from that literature.}

\section{Quantifying and Maximizing the Robustness Considering Load/Resource Fluctuations\label{sec:fluctuations}}
\subsection{Uniform resource/load fluctuations}\label{hom}

For the uniform resource fluctuations, we pursue the worst-case design approach assuming $R_k\rightarrow R_k-\delta$, $\forall s_k \in \mathcal{S}$, $\delta\in \mathbb{R}^+$ (see Fig.~\ref{diag:show2} (b)). To study the stability of the supply nodes, we define the \textit{free capacity} $C_k$ of the supply node $s_k\in\mathcal{S}$ as:
\begin{equation}\label{eq:freecapac}
C_k=R_k-r_k.
\end{equation} 
We measure the robustness of the network in terms of the maximum tolerability of the network against resource fluctuations $MTRF$ defined as:
\begin{equation}
\begin{aligned}
MTRF&=\displaystyle \min_{ \delta\in \mathbb{R}^+} \{\delta: C_k-\delta \leq 0, \;\;\exists s_k\in \widehat{\mathcal{S}} \}\\
&\equiv \min  \{C_k:s_k\in \widehat{\mathcal{S}}\} ,
\end{aligned}
\end{equation}
where $\Scale[0.95]{\widehat{\mathcal{S}}\triangleq\{s_k: r_k>0\}}$ denotes the set of supply nodes involved in resource provisioning. In other words, $MTRF$ is equal to the minimum amount of resource fluctuation that results in the instability of at least a supply node engaged in resource provisioning, which is equivalent to the minimum free capacity among the supply nodes offering resources to the demand layer. It can be construed that, upon uniform resource fluctuations, an increase in this parameter is directly linked to prevention of cascading failures. In the described demand-supply bipartite graph, we consider the weight of an edge connecting a pair of demand and supply nodes to be the amount of shared resources between them. 

Our goal is then to find the optimal weighted adjacency matrix $\mathbf{P}^*=[\rho_{(s_k,d_g)}]_{s_k\in \mathcal{S}, d_g \in \mathcal{D}}$ corresponding to a network with the maximum $MTRF$. One approach to achieve the largest $MTRF$ is to minimize the offered resources used at each supply node, which leads to the maximum free capacity for each of them. On the other hand, simultaneously satisfying the loads of the demand nodes and minimizing the offered resources used at each supply node is equivalent to stabilizing the demand nodes by satisfying $l_g = L_g, \;\forall d_g \in \mathcal{D}$ (see~\eqref{stabledemand}). Using these facts, the resource configuration $\mathbf{P}^*=[\rho_{(s_k,d_g)}]_{s_k\in \mathcal{S}, d_g \in \mathcal{D}}$ achieving the highest robustness against uniform resource fluctuations can be obtained by solving the following optimization problem:
\begin{flalign} &\mathbf{P}^*=\argmax_{\mathbf{P}=[\rho_{(s_k,d_g)}]_{s_k\in \mathcal{S}, d_g \in \mathcal{D}}}\{\min_{s_i\in \hat{\mathcal{S}}} \{C_i\}\}\label{eq:dummy} \\
(\bm{\mathcal{P}1}) \hspace{12mm}&\textrm{s.t.}\nonumber \\
&r_i \leq R_i,\;\;  \forall s_i \in \mathcal{S}, \label{eq:c1} \\
&l_g = L_g,\;\;\forall d_g \in \mathcal{D},\label{eq:c2}\\
&\rho_{(s_i,d_g)}\geq 0,\; \;\forall d_g \in \mathcal{D},\;  \forall s_i \in \mathcal{S},\label{eq:c3}
\end{flalign}
where~\eqref{eq:c1} implies the stability of the supply nodes,~\eqref{eq:c2} ensures the functionality of the demand nodes, and~\eqref{eq:c3} guarantees a feasible allocation. In the following, we aim to solve this problem. Note that the objective function of $\bm{\mathcal{P}1}$ is not a strict concave function. Hence, there can be multiple resource configuration schemes that are the solutions to $\bm{\mathcal{P}1}$. As further discussed in Remark~\ref{rem1} and Section~\ref{sec:cost-based} below, identifying a unique weighted adjacency matrix requires the knowledge of the cost of resource sharing among the nodes. The following fact serves as the basis to solve $\bm{\mathcal{P}1}$.
\begin{fact}
If the sum of the offered resource to the demand layer is larger than or equal to the sum of the loads of the demand nodes while all the supply nodes are stable, there is at least a resource configuration scheme that leads to the stability of all the nodes.
\end{fact}
With the above fact, we first focus on obtaining the optimal resource that supply nodes need to offer to the demand layer. The derivation of the optimal weighted adjacency matrix, i.e., the optimal resource configuration, is deferred to Section~\ref{sec:cost-based}.

\begin{lemma}\label{cor:1}
Assume that the resources of the supply nodes can be sorted as $R_{i_1}> R_{i_2}> \cdots > R_{i_S}$, $i_k \in \{1,\cdots,S \},\; \forall k$, and define the dummy variable $R_{i_{S+1}}=0$. Let $v^*$ denote the smallest index in which $\sum_{v=1}^{v^*} R_{i_v} - v^* R_{i_{v^*+1}} \geq \sum_{g=1}^{D} L_g$. Then, $\sum_{g=1}^{D}\rho^*_{(s_k,d_g)}>0 \; \forall k \in \{i_1,\cdots, i_{v^*}\}$, and $ \sum_{g=1}^{D} \rho^*_{(s_k,d_g)}=0 \; \forall k \in \{i_{v^*+1},\cdots, i_{S}\}$.
\end{lemma}
\begin{proof}
Please refer to the supplementary materials.
\end{proof}
\begin{remark}
The above lemma determines the resource nodes involved in the resource provisioning: $\widehat{\mathcal{S}}=\{s_{i_1},s_{i_2},\cdots,s_{i_{v^*}} \}$. The main intuition behind the lemma is the fact that the $MTRF$ is defined as the minimum free capacity of the resource nodes involved in resource provisioning. Hence, to maximize this metric those nodes with higher resources should be used. More explanation on this fact along with the examination of the case in which resource nodes may have equal amout of resources are given in the supplementary materials.
\end{remark}

\begin{theorem}
For the sorted resources and $v^*$ described in Lemma~\ref{cor:1}, the optimal resource offered from supply nodes is given by:
\begin{equation}\label{eq:firstAns4}
\begin{cases}
\sum_{g=1}^{D} \rho^*_{(s_k,d_g)}= R_k -\frac{\sum_{j=1}^{v^*} R_{i_j} - \sum_{g=1}^{D} L_g}{v^*}, \\ 
\;\forall k \in \{i_1,\cdots, i_{v^*}\}, \\ 
\rho^*_{(s_k,d_g)}=0 \;\forall k \in \{i_{v^*+1},\cdots, i_{S}\},\; \forall d_g \in \mathcal{D}.
\end{cases}
\end{equation}
\end{theorem}
\vspace{1.5mm}
\begin{proof}
Assume that the set of resource nodes involved in the resource provisioning, according to Lemma~\ref{cor:1}, are given by: $\widehat{\mathcal{S}}=\{s_{i_1},s_{i_2},\cdots,s_{i_{v^*}} \}$. Upon solving $\bm{\mathcal{P}1}$, assume that the free capacities of these nodes after resource provisioning  are sorted as $C_{j_1}\leq C_{j_2}\leq \cdots \leq C_{j_{\nu^*}}$. Define the dummy variables $G_k=C_{j_k}-C_{j_1}$, $\forall k$. Since $l_g=L_g$, $\forall d_g\in \mathcal{D}$, we get $\sum_{k=1}^{\nu^*} C_{j_k}=\sum_{k=1}^{\nu^*} R_k-\sum_{k=1}^{D} L_k$. On the other hand, $\sum_{k=1}^{\nu^*} G_k = \sum_{k=1}^{\nu^*} C_{j_k}- \nu^* C_{j_1}$. Using these facts, the objective function of $\bm{\mathcal{P}1}$ can be written as:
\begin{equation}
\begin{aligned}
    &\hspace{-3mm}\argmax_{\mathbf{P}=[\rho_{(s_k,d_g)}]_{s_k\in \mathcal{S}, d_g \in \mathcal{D}}}\{\min_{s_i\in \hat{\mathcal{S}}} \{C_i\}\}\\
   &\hspace{-3mm} \equiv\hspace{-3mm} \argmax_{\mathbf{P}=[\rho_{(s_k,d_g)}]_{s_k\in \mathcal{S}, d_g \in \mathcal{D}}} C_{j_1} \\
    &\hspace{-3mm}\equiv\hspace{-3mm}
    \argmax_{\mathbf{P}=[\rho_{(s_k,d_g)}]_{s_k\in \mathcal{S}, d_g \in \mathcal{D}}} \left(\sum_{k=1}^{\nu^*} C_{j_k} -\sum_{k=1}^{\nu^*} G_k\right)/\nu^*\\
    &\hspace{-3mm}\equiv\hspace{-3mm}
    \argmax_{\mathbf{P}=[\rho_{(s_k,d_g)}]_{s_k\in \mathcal{S}, d_g \in \mathcal{D}}} \left(\sum_{k=1}^{\nu^*} R_k-\sum_{k=1}^{D} L_k -\sum_{k=1}^{\nu^*} G_k\right)/\nu^*.
    \end{aligned}
    \hspace{-6mm}
\end{equation}
From the last expression, it can be verified that the optimal solution is attained when $\sum_{k=1}^{\nu^*} G_k=0$. Since $G_k\geq 0$, this is achieved when $G_k=0$, $\forall k$. This is equivalent to $C_{j_k}=C_{j_1}$, $\forall k$, implying the same free capacity for those supply nodes involved in resource provisioning. Considering this result and~\eqref{eq:c2}, we get: $\sum_{s_i\in \hat{\mathcal{S}}} C_i= \sum_{s_i\in \hat{\mathcal{S}}} R_i-\sum_{g=1}^{D} L_g$, which implies $\nu^* C_i= \sum_{s_k\in \hat{\mathcal{S}}} R_k-\sum_{g=1}^{D} L_g$, $\forall s_i\in \hat{\mathcal{S}}$. Considering~\eqref{eq:freecapac}, we get: $r_i= R_i-\left(\sum_{s_k\in \hat{\mathcal{S}}} R_k-\sum_{g=1}^{D} L_g\right)/\nu^*$, $\forall s_i\in \hat{\mathcal{S}}$, and the theorem is proved.
\end{proof}

So far, the resource fluctuations was the main concern of the design. Nevertheless, achieving a high robustness against uniform load fluctuations is a straightforward extension. We define the maximum tolerable load fluctuations $MTLF$ as:
\begin{flalign}
MTLF=\displaystyle{\min_{\delta'}}\{{\delta'}:C_k - \frac{\delta'}{|\mathcal{N}(d_g)|}\leq 0,\nonumber\\
\exists d_g \in \mathcal{D}, \exists s_k\in \mathcal{N}(d_g)\},
\end{flalign}
where $\mathcal{N}(d_g)$ denotes the set of adjacent supply nodes of $d_g$, $\mathcal{N}(d_g)\triangleq\{s_k: s_k\in \mathcal{S}, \rho_{(s_k,d_g)}>0 \}$. In other words, $MTLF$ is the minimum value of increment in the load of demand nodes, which results in failure of (at least) one supply node under uniform spreading of the resource deficiencies. Using~\eqref{eq:firstAns4}, it can be construed that the free capacities of all the supply nodes with non-zero offering resources to the demand layer is the same $C_i=\left(\sum_{j=1}^{v^*} R_{i_j} - \sum_{g=1}^{D} L_g\right)/v^*$, $\forall s_i \in \widehat{\mathcal{S}}$. In this case, maximizing the value of $|\mathcal{N}(d_g)|$, i.e., connecting each demand node to all the supply nodes, leads to the maximization of $MTLF$. For this purpose, after solving $\mathcal{P}1$, the designer needs to consider a small value $\epsilon$ and allocate $\epsilon$ unit of resources from each supply node involved in the resource provisioning, i.e., $\forall s_k\in \widehat{\mathcal{S}}$, to all the demand nodes: $\rho^*_{(s_k,d_g)}> \epsilon, \forall d_g \in \mathcal{D}, \forall s_k\in \widehat{\mathcal{S}}$. 

   \begin{remark}\label{rem1}
Note that~\eqref{eq:firstAns4} determines the total offered resources from each supply node without specifying the exact values of resources shared among different pairs of nodes. This is due to focusing on the free capacities of the nodes while assuming no difference in terms of resource sharing through different links. Hence, among the optimal solutions, any resource distribution strategy satisfying $\sum_{i=1}^{S} \rho^*_{(s_i,d_g)} = L_g,\; \forall d_g \in \mathcal{D}$, can be deployed. Upon having heterogeneous costs associated with different links, the precise values of shared resources among the nodes can be derived (see Section~\ref{sec:cost-based}). The same philosophy holds for the discussions in Section~\ref{het}.
\end{remark}
\subsection{Proportional resource/load fluctuations}\label{het}
To capture the proportional distribution of load fluctuations among the supply nodes, we model the offered resources of supply node $s_i\in \mathcal{S}$ to demand node $d_g \in \mathcal{D}$ as $\rho_{(s_i,d_g)}=w_{ig}L_g$, where $w_{ig} \in [0,1]$ denotes the weight of the edge between them determining the fraction of the load of $d_g$ offered via $s_i$. In this case, the condition $\sum_{i=1}^{S} w_{ig} = 1$ ensures the functionality of demand node $d_g \in \mathcal{D}$. We design the system for the worst-case scenario, where the load of each demand node increases from $L_g$ to $\Xi' L_g$, $\forall d_g \in \mathcal{D}$,  $\Xi'>1$. In this case, to maintain the stability of the demand node $d_g \in \mathcal{D}$, supply node $s_i \in \mathcal{S}$ should provide $w_{ig} (\Xi'-1) L_g$ extra resources to the demand node. Hence, the requested resource that each supply $s_i \in \mathcal{S}$ should provide changes as $r_i \rightarrow \Xi' r_i$. Also, in resource fluctuations scenario, the resources of each supply node decreases from $R_i$ to $(1-\Xi) R_i$, $\forall s_i \in \mathcal{S}$, where $0<\Xi<1$. For this scenario, we define the $MTLF$ and the $MTRF$ as follows: 
\begin{align}
&MTLF=\displaystyle \max_{\Xi'>1} \{\Xi': \Xi' r_i\leq   R_i,\;\; \forall s_i \in \widehat{\mathcal{S}} \}, \\
&MTRF=\displaystyle \max_{0<\Xi<1} \{\Xi: (1-\Xi)R_i\geq r_i, \;\;\forall s_i \in \widehat{\mathcal{S}}  \}.
\end{align}
In other words, in this case, $MTLF$ and $MTRF$ are the maximum amount of proportional load increment and the maximum amount of proportional resource reduction, for which the stability of all the supply nodes is guaranteed. In the following lemma and theorem, we identify the maximum attainable value of these parameters and propose a simple resource configuration scheme, which simultaneously achieves the maximum attainable values of both. 
\vspace{1.5mm}
\begin{lemma}
Given a demand-supply network, for any resource configuration scheme, $MTLF$ and $MTRF$ are bounded as follows: $MTLF\leq \frac{\sum_{s=1}^{S} R_s }{\sum_{g=1}^{D}L_g}, MTRF\leq 1- \frac{\sum_{g=1}^{D}L_g}{\sum_{s=1}^{S} R_s}$.
\end{lemma}
\vspace{1.5mm}
\begin{proof}
The result can be derived using the proof by contradiction on the definitions of $MTLF$ and $MTRF$ considering the free capacities of the supply nodes.
\end{proof}
\vspace{1.5mm}
\begin{theorem}
Any resource configuration $W^*=[w^*_{ig}]_{s_i\in \mathcal{S}, d_g\in \mathcal{D}}$ satisfying the following criteria results in a network with the highest attainable values of both $MTLF$ and $MTRF$:
\begin{equation}\label{eq:dum}
\sum_{g=1}^{D} w^*_{ig}L_g = \frac{R_{i}}{\sum_{j=1}^{S}R_{j}} \sum_{g=1}^{D} L_g \; \forall s_i\in \mathcal{S}.
\end{equation}
\end{theorem}
\begin{proof}
Considering~\eqref{eq:dum},
$C_i=R_i-\sum_{g=1}^{D} w^*_{ig}L_g=R_i-\frac{R_{i}}{\sum_{j=1}^{S}R_{j}} \sum_{g=1}^{D} L_g$, $\forall s_i\in \mathcal{S}$. By increasing the loads from $\sum_{g=1}^{D}L_g$ to $\sum_{g=1}^{D}\Xi' L_g$, where $\Xi'>1$, the free capacity is given by $\hat{C}_i=R_i-\frac{\Xi' R_{i}}{\sum_{j=1}^{S} R_{j}} \sum_{g=1}^{D} L_g$, $\forall s_i\in \mathcal{S}$. The $MTLF$ can be derived by solving $C_i=0,\; \forall s_i\in \mathcal{S}$, which results in $ MTLF= \frac{\sum_{s=1}^{S} R_s }{\sum_{g=1}^{D}L_g}$. Also, $MTRF$ can be obtained considering the free capacity of the supply nodes after resource decrement, $\hat{C}_i=(1-\Xi) R_i-\frac{R_{i}}{\sum_{j=1}^{S}R_{j}} \sum_{g=1}^{D} L_g$, solving $\hat{C}_i=0,\; \forall s_i\in \mathcal{S}$, results in $ MTRF=1-\frac{\sum_{g=1}^{D}L_g}{\sum_{s=1}^{S} R_s}$.
\end{proof}

\section{Resource Configuration under Heterogeneous Geographical Constraints}\label{sec:cost-based}
Consider a scenario in which allocation of resources among different pairs of demand and supply nodes is associated with different costs. The cost can represent geographical constraints or, in general, heterogeneous tendencies of supply nodes toward resource allocation among different demand nodes. We model the allocation cost between $s_i \in \mathcal{S}$ and $d_g\in \mathcal{D}$ with corresponding allocated resource $\rho_{(s_i,d_g)}$ as:
\begin{equation}\label{eq:cost}
    C(\rho_{(s_i,d_g)} )=\alpha_{ig}\left(\rho_{(s_i,d_g)} \right)^{\beta_{ig}},
\end{equation}
where $\alpha_{..}\in \mathbb{R}^+$ and $\beta_{..} \in (1,+\infty) $ are node specific parameters capturing the inherent heterogeneities. Note that our approach can be applied to any convex and increasing cost function and the cost modeling here is just for mathematical convenience.\footnote{The results provided in this section are revisited for another general family of cost functions in Section 3 of the supplementary materials.} In the following, we first obtain the precise values of allocated resources among the nodes to achieve the highest robustness while minimizing the allocation cost. Afterward, we provide an algorithm to reduce the allocation cost of a given network via re-adjusting the resource configuration.
\subsection{Achieving the highest robustness while minimizing the allocation cost}
In~\eqref{eq:firstAns4} and  \eqref{eq:dum} the total offered resources from each supply node is derived to achieve the highest robustness without specifying the exact resource configuration between each pair of supply and demand node. Considering the heterogeneous costs of resource allocation among the nodes, to identify the cost-effective resource configuration with the highest robustness, we propose the following optimization problem, which determines the precise value of allocated resources for all pairs of nodes $\widetilde{\mathbf{P}^*}=[\widetilde{\rho^*}_{(s_i,d_g)}]_{s_i\in \mathcal{S}, d_g \in \mathcal{D}}$: 
\vspace{1.5mm}
\begin{flalign}
   &\widetilde{\mathbf{P}^*}=\argmin_{\widetilde{\mathbf{P}}=[\widetilde{\rho}_{(s_k,d_g)}]_{s_k\in \mathcal{S}, d_g \in \mathcal{D}}} \sum_{s_i\in \mathcal{S}} \sum_{d_j \in D} C(\widetilde{\rho}_{(s_i,d_g)} )
   \\
   (\bm{\mathcal{P}2})\hspace{3mm} &\textrm{s.t.}\nonumber \\
   &\sum_{g=1}^{D}\widetilde{\rho}_{(s_i,d_g)}= \sum_{g=1}^{D}\rho^*_{(s_i,d_g)},\;\;  \forall s_i \in \mathcal{S},\label{eq:41}\\
   &\sum_{i=1}^{S} \widetilde{\rho}_{(s_i,d_g)} =L_g,\;\;  \forall d_g \in \mathcal{D}, \label{eq:42} \\
   &\widetilde{\rho}_{(s_i,d_g)} \geq 0, \;\;\forall s_i \in \mathcal{S},d_g \in \mathcal{D},\label{eq:43} 
\end{flalign}
where the right hand side of the first constraint denotes the solutions obtained from the previous section for either uniform or proportional load/resource fluctuations scenario (\eqref{eq:firstAns4}, \eqref{eq:dum}). To obtain the solution, consider the Lagrangian multipliers associated with the first, the second, and the third constraint as $\bm{\Lambda}=\left[\lambda_1,\lambda_2,\cdots,\lambda_S \right]$, $\mathbf{Z}=\left[\zeta_1,\zeta_2,\cdots,\zeta_D \right]$, and $\bm{\Phi}=\left[\phi_{ig}\right]_{s_i\in \mathcal{S}, d_g \in \mathcal{D}}$, respectively. The corresponding Lagrangian function of $\bm{\mathcal{P}2}$ is given by:
\vspace{1.5mm}
\begin{equation}
\begin{aligned}
   & L(\mathbf{P},\bm{\Lambda},\mathbf{Z},\bm{\Phi})=\sum_{s_i\in \mathcal{S}} \sum_{d_g \in D}\alpha_{ig}\left(\widetilde{\rho}_{(s_i,d_g)} \right)^{\beta_{ig}}\\\
   &+\sum_{s_i \in \mathcal{S}} \lambda_i \left(\sum_{d_g\in D} \widetilde{\rho}_{(s_i,d_g)} -X^*_i\right)\\
   &+\sum_{d_g\in D} \zeta_g \left(\sum_{s_i \in \mathcal{S}} \widetilde{\rho}_{(s_i,d_g)} -L_g\right) \\
   &- \sum_{s_i\in \mathcal{S}} \sum_{d_g \in D}\phi_{ig}\widetilde{\rho}_{(s_i,d_g)},
    \end{aligned}
\end{equation}
where $X^*_i= \sum_{g=1}^{D}\rho^*_{(s_i,d_g)}$.
From the \textit{stationarity} condition of KKT conditions~\cite{ConvexOptimization}, we get:
\begin{flalign}
    \partial L(\mathbf{P},\bm{\Lambda},\mathbf{Z},\bm{\Phi})/\partial \widetilde{\rho}_{(s_i,d_g)}=\alpha_{ig}\beta_{ig}\left(\widetilde{\rho}_{(s_i,d_g)})\right)^{\beta_{ig}-1}\nonumber \\
    +\lambda_i+\zeta_g-\phi_{ig}=0,
    \end{flalign}
which results in:
\vspace{1.5mm}
\begin{equation}
  \widetilde{ \rho^*}_{(s_i,d_g)}=  \left(\frac{-\lambda_i-\zeta_g+\phi_{ig}}{\beta_{ig}\alpha_{ig}}\right)^{\frac{1}{\beta_{ig}-1}}.
\end{equation}
As can be seen, $\widetilde{ \rho^*}(.,.)$ is a function of the Lagrangian multipliers, the values of which can be found by applying the gradient descent method~\cite{ConvexOptimization} on the following dual problem:
\begin{equation}
\max_{\bm{\Lambda} \in \mathbb{R}^{S},\mathbf{Z} \in \mathbb{R}^{D}, \bm{\Phi}\in \left( \mathbb{R}^+ \right)^{S\times D}} D(\bm{\Lambda},\mathbf{Z},\bm{\Phi}),
\end{equation}
where the dual function $D(\bm{\Lambda},\mathbf{Z},\bm{\Phi})$ is given by:
\begin{flalign}
   \hspace{-1mm} D(&\bm{\Lambda},\mathbf{Z},\bm{\Phi})= \sum_{s_i\in \mathcal{S}} \sum_{d_g \in D}\alpha_{ig}\left(\frac{-\lambda_i-\zeta_g+\phi_{ig}}{\beta_{ig}\alpha_{ig}}\right)^{\frac{\beta_{ig}}{\beta_{ig}-1}}\nonumber\\
   &+\sum_{s_i \in \mathcal{S}} {\lambda_i} \left(\sum_{d_g\in D} \left(\frac{-\lambda_i-\zeta_g+\phi_{ig}}{\beta_{ig}\alpha_{ig}}\right)^{\frac{1}{\beta_{ig}-1}} -X^*_i\right)\nonumber\\
   &+\sum_{d_g\in D} {\zeta_g} \left(\sum_{s_i \in \mathcal{S}} \left(\frac{-\lambda_i-\zeta_g+\phi_{ig}}{\beta_{ig}\alpha_{ig}}\right)^{\frac{1}{\beta_{ig}-1}} -L_g\right)\nonumber\\
   &- \sum_{s_i\in \mathcal{S}} \sum_{d_g \in D}\phi_{ig}\left(\frac{-\lambda_i-\zeta_g+\phi_{ig}}{\beta_{ig}\alpha_{ig}}\right)^{\frac{1}{\beta_{ig}-1}}.
    \end{flalign}   
Since $\bm{\mathcal{P}2}$ is a convex optimization problem with affine constraints, the optimal solution of the dual problem will coincide with the optimal solution of the original problem. 
	\begin{algorithm}[t]
	{\footnotesize
		\caption{An iterative algorithm for cost reduction while incurring the least loss in robustness}\label{alg:2}
		\SetKwFunction{Union}{Union}\SetKwFunction{FindCompress}{FindCompress}
		\SetKwInOut{Input}{input}\SetKwInOut{Output}{output}
		\Input{The network configuration $\rho_{(s_k,d_g)} \; \forall s_k \in \mathcal{S}, d_g \in \mathcal{D}$, Reduction step size $\epsilon$, Objective Cost $C^O$}
		\Output{The load allocation configuration $\rho_{(s_k,d_g)} \;  \forall s_k \in \mathcal{S}, d_g \in \mathcal{D}$}
		Calculate the cost: $C^I=\displaystyle\sum_{k=1}^{S}\displaystyle \sum_{g=1}^{D}\alpha_{kg}\left(\rho_{(s_k,d_g)} \right)^{\beta_{kg}}$\\
		$\Delta_C=C^I-C^O$\\
		\While{$\Delta_C>0$}{
			Calculate  $C'_{kg}=\frac{\partial \alpha_{kg}\left(\rho_{(s_k,d_g)} \right)^{\beta_{kg}} }{\partial \rho_{(s_k,d_g)}} \;  \forall s_k \in \mathcal{S}, d_g \in \mathcal{D}$\\
			$(k^*,g^*)=\displaystyle\argmax_{1\leq k \leq S, 1\leq g\leq D} C'_{kg}$ (break ties uniformly at random)\\
			Among the set of solutions of $(\tilde{k},\tilde{g})=\displaystyle\argmin_{\forall s_k \in \mathcal{S}, d_g \in \mathcal{D}} C'_{kg}$ find the supply node $s_{\hat{k}^*}$ as follows (break ties uniformly at random):
			\begin{equation}\label{eq:midAlg}
			\begin{cases}
		\textrm{Uniform}:	\hat{k}^* = \displaystyle\argmax_{ \hat{k} \in \tilde{k} } R_{\hat{k}}-\sum_{g=1}^{D} \rho_{(s_{\hat{k}},d_g)},\\
		\textrm{Proportional}:    \hat{k}^* = \displaystyle\argmax_{ \hat{k} \in \tilde{k}}  R_{\hat{k}}/ \sum_{g=1}^{D} \rho_{(s_{\hat{k}},d_g)}.
			\end{cases}
			\end{equation}\\
           Tune the reduction size: $\kappa=\min(\epsilon,\rho_{(s_{k^*},d_{g^*})},C_{\hat{k}^*})$\\
		   Re-adjust the resource: $\rho_{(s_{k^*},d_{g^*})}= \rho_{(s_{k^*},d_{g^*})}-\kappa$\\
		   Re-adjust the resource: $\rho_{(s_{\hat{k}^*},d_{g^*})}= \rho_{(s_{\hat{k}^*},d_{g^*})}+\kappa$\\
		   $C^I=\displaystyle\sum_{k=1}^{S}\displaystyle \sum_{g=1}^{D}\alpha_{kg}\left(\rho_{(s_k,d_g)} \right)^{\beta_{kg}}$\\
		   \vspace{1.5mm}
		   $\Delta_C=C^I-C^O$\\	
		}
		
	}
\end{algorithm}	
\subsection{Cost reduction while incurring the least loss in robustness}\label{sec:cost-cosntraitn}

For a given network, the network operator may need to reduce the resource allocation cost due to economical situations or upon change of link cost parameters among the nodes. In this situation, it is more reasonable to reduce the allocation cost by locally re-adjusting the resources offered through a portion of links as compared to (re-)designing the entire network. Moreover, to prevent cascading failures from initiation in the final network, it is desired to  incur the least loss in the robustness while performing the resource re-adjustment. Let $C^I$ and $C^O$ denote the allocation cost of the initial network and the objective cost, respectively, and $\Delta_C=C^I-C^O$. We propose an iterative resource re-adjustment method, described in Algorithm~\ref{alg:2}, aiming to reduce the maximum attainable cost at each iteration while imposing the least impact on the robustness. At each iteration of this algorithm, first the partial derivatives of the link costs with respect to the allocated resources are obtained. Then, the resources allocated to the link with the highest derivative is decreased by a step-size variable $\epsilon\in \mathbb{R}^+$. Note that for a given step-size variable $\epsilon\in \mathbb{R}^+$, the resource reduction through that link is associated with the maximum attainable cost reduction among all the links.  Afterward, the emerged resource deficiency of the associated demand node to that link is compensated through allocating extra resources from the supply node with the lowest associated cost derivative that has the highest tolerance of resource fluctuations  (uniform or proportional depending on the context) obtained from~\eqref{eq:midAlg}. As can be construed, this resource compensation procedure ensures the smallest impact on the robustness of the network.
\subsubsection{On the feasibility of cost reduction}
Considering the resource satisfaction paradigm (\eqref{stableSupply},~\eqref{stabledemand}), our algorithm is not capable of achieving any arbitrary desired cost $C^O$. This parameter is lower-bounded with respect to the link cost parameters and the feasible resource configuration schemes. In the following proposition, we derive the minimum achievable resource sharing cost. 
\vspace{1.5mm}
\begin{proposition}
Assume $\bm{\Lambda}=\left[\lambda_1,\lambda_2,\cdots,\lambda_S \right]$, $\mathbf{Z}=\left[\zeta_1,\zeta_2,\cdots,\zeta_D \right]$, and $\bm{\Phi}=\left[\phi_{ig}\right]_{s_i\in \mathcal{S}, d_g \in \mathcal{D}}$. The minimum attainable budget for successful cost reduction using Algorithm~\ref{alg:2} $C^O_{min}$ is given by $C^O_{min}=\sum_{s_i\in \mathcal{S}} \sum_{d_g \in D}\alpha_{ig}\left(\widetilde{\rho^*}_{(s_i,d_g)} \right)^{\beta_{ig}}$, where $
  \widetilde{ \rho^*}_{(s_i,d_g)}= \left(\frac{-\lambda_i-\zeta_g+\phi_{ig}}{\beta_{ig}\alpha_{ig}}\right)^{\frac{1}{\beta_{ig}-1}}$,
for which the value of the coefficients can be found by applying the gradient descent method on the following optimization problem:\footnote{The mentioned value for $C^O_{min}$ is in fact the minimum achievable cost of allocation that guarantees the stability of the nodes.}
\vspace{1.5mm}
\begin{equation}
[\bm{\Lambda^*},\mathbf{Z^*}, \bm{\Phi^*}]=\argmax_{\bm{\Lambda} \in \mathbb{(R^+)}^{S},\mathbf{Z} \in \mathbb{R}^{D}, \bm{\Phi}\in \left( \mathbb{R}^+ \right)^{S\times D}} U(\bm{\Lambda},\mathbf{Z},\bm{\Phi}),
\end{equation}
where
\vspace{1.5mm}
\begin{flalign}
   \hspace{6mm} U(&\bm{\Lambda},\mathbf{Z},\bm{\Phi})= \sum_{s_i\in \mathcal{S}} \sum_{d_g \in D}\alpha_{ig}\left(\frac{-\lambda_i-\zeta_g+\phi_{ig}}{\beta_{ig}\alpha_{ig}}\right)^{\frac{\beta_{ig}}{\beta_{ig}-1}}\nonumber\\
   &+\sum_{s_i \in \mathcal{S}}  {\lambda_i}\left(\sum_{d_g\in D} \left(\frac{-\lambda_i-\zeta_g+\phi_{ig}}{\beta_{ig}\alpha_{ig}}\right)^{\frac{1}{\beta_{ig}-1}} -R_i\right)\nonumber\\
   &+\sum_{d_g\in D} {\zeta_g} \left(\sum_{s_i \in \mathcal{S}}  \left(\frac{-\lambda_i-\zeta_g+\phi_{ig}}{\beta_{ig}\alpha_{ig}}\right)^{\frac{1}{\beta_{ig}-1}} -L_g\right)\nonumber\\
   &- \sum_{s_i\in \mathcal{S}} \sum_{d_g \in D}\phi_{ig}\left(\frac{-\lambda_i-\zeta_g+\phi_{ig}}{\beta_{ig}\alpha_{ig}}\right)^{\frac{1}{\beta_{ig}-1}}.
    \end{flalign}  
\end{proposition}
\begin{proof}
Please refer to the supplementary materials.
\end{proof}
\section{Stopping/confining Cascading Failures while Maintaining a High Robustness}\label{sec:confine}
\noindent  So far, the design goal was to prevent cascading failures from happening. In this section, we aim to develop an adaptation scheme tailored for demand-supply networks to confine ongoing cascading failures in the network. In particular,
considering any arbitrary demand-supply network, once cascading failures start to propagate, we aim to achieve the following goals simultaneously: i) stopping/confining the ongoing cascading failures, ii) providing the surviving network with a high robustness to avoid further failures. To achieve this, we introduce and deploy two network adaptation mechanisms: \textit{intentional failure} and \textit{resource re-adjustment}. In the intentional failure mechanism, we deliberately cut all the edges connecting to the demand node with the largest resource deficiency, making it isolated from the supply layer. This can be thought as sacrificing a demand node to (potentially) save the rest of the network. The resource re-adjustment mechanism is composed of the following parts. i)~We compensate for the resource deficiencies of the demand nodes via utilizing the spare resources of the supply nodes with the highest tolerance of resource fluctuations in an iterative manner. This procedure ensures the least impact on the robustness of the network while compensating for resource deficiencies. ii) By re-allocating the offered resources from the supply nodes with low tolerance of fluctuations to those with high tolerance (larger capacities), we modify the resource configuration scheme to achieve a higher robustness. This stage is performed to balance the free capacities of the nodes in the surviving network, which results in a higher robustness. As compared to the current state-of-the-art in adaptability of interdependent networks, e.g., \cite{RecoveryOfInterdependentNetworks,StrategyForStoppingFailureCascadesInInterdependentNetworks}, which mainly focus on rehabilitating the failed nodes to confine cascading failures, our approach is fundamentally different due to the distinct cascading failure mechanism in demand-supply networks. Our approach can be recognized as \textit{smart failing} of the nodes and \textit{sequential resource re-adjustment}, which simultaneously confine cascading failures and enhance the robustness. 
\begin{algorithm}[t]
		{\footnotesize
			\caption{An iterative algorithm for confining a cascading failure while increasing the robustness of the surviving network}\label{alg:3}
			\SetKwFunction{Union}{Union}\SetKwFunction{FindCompress}{FindCompress}
			\SetKwInOut{Input}{input}\SetKwInOut{Output}{output}
			\Input{The network configuration $\rho_{(s_k,d_g)}(0) \forall s_k\in\mathcal{S},d_g\in\mathcal{D}$, the set of failed demand nodes at time zero $\mathcal{F}_s(0)$, maximum number of intentional failures $\Gamma$, maximum amount of resource re-adjustments $\Upsilon$, load and resources of the nodes.}
			t=0\\
			Derive the set of operational supply nodes $\mathcal{O}_S(t)= \mathcal{S}\setminus \mathcal{F}_s(t)$\label{line3alg3}.\\
			Derive the set of operable demand nodes before the failure $\mathcal{O}_D(t)$.\\
			Derive the set of demand nodes with resource deficiency ${\Delta}_D(t)$.\\
            Let $\delta_g$ denote the amount of resource deficiency at $d_g(t) \in {\Delta}_D(t)$.\\
			Let the sequence $\{d_{i_k}\}_{k=1}^{|{\Delta}_D(t)|}$ denote the elements of ${\Delta}_D(t)$ sorted in descending order with respect to the amount of resource deficiency.  \\
			$\Xi(t)=\sum_{s_i\in \mathcal{O}_s(t)} C_i$\\
			$\gamma=1$\\
           \nonl $\backslash \backslash$ Perform intentional failures:\\
			\While{$\gamma \leq \Gamma$ \textrm{and} $\sum_{d_k\in \mathcal{O}_D(t)} d_g(t) > \Xi(t)$}{
				
					$\mathcal{O}_D(t)=\mathcal{O}_D(t) \setminus d_{i_{\gamma}} $\\
                    ${\Delta}_D(t)={\Delta}_D(t) \setminus d_{i_{\gamma}}$\\
					$\gamma=\gamma+1$
			}
            \nonl $\backslash \backslash$ Check for the feasibility of satisfying the resource deficiencies:\\
		\If{$\sum_{d_k\in \mathcal{O}_D(t)} d_g(t) \leq \Xi(t)$ and $\sum_{d_k\in {\Delta}_D(t)} \delta_g(t) \leq \Upsilon$}{
			Perform resource re-adjustment using Algorithm~\ref{alg:4}.
		}\Else{
		Let the cascading failure continue for one time instant and obtain  $\mathcal{F}_s(t+1)$, $\mathcal{O}_D(t+1)$, and ${\Delta}_D(t+1)$ \\
		$t=t+1$\\
		go to line \ref{line3alg3}.	
	}
						
		}
		\vspace{-1mm}
	\end{algorithm}	  
We consider a realistic scenario in which the network operator has a limited capability/ability to perform network adaptation mechanisms. In particular, we address the network operator's capability by two parameters $\Gamma \in \mathbb{Z}^+ \cup \{0\}$ and $\Upsilon\in \mathbb{R}^+\cup \{0\}$ indicating the number of intentional failures and the total amount of resource re-adjustments that can be executed at each time instant, respectively. The pseudo-code of our proposed network adaptation algorithm is given in Algorithm~\ref{alg:3}. At each stage of cascading failures, the algorithm first identifies the best set of candidates for intentional failures, which consists of those demand nodes with the highest resource deficiencies. Failing those demand nodes results in the minimum required amount of extra resource to stabilize the network. Afterward, if it predicts that the cascading failures can be stopped using the spare resources of the supply nodes, i.e., the resource deficiencies can be fulfilled via feasible resource re-adjustments, it performs the resource re-adjustments to stop cascading failures and continues the re-adjustments to balance the free capacities of the supply nodes using Algorithm~\ref{alg:4}. Note that at the current stage of cascading failures, if the algorithm predicts that it cannot confine the failures, it waits until next time instant to (possibly) handle less amount of resource deficiencies due to potential failures.

In the end, there are two worthwhile remarks regarding our proposed algorithm: i) The resource re-adjustment stage of the algorithm can be applied by itself to any demand-supply network to increase the robustness. ii) Instead of improving the robustness, after the cascading failures stop, the resource re-adjustment can also be implemented from an alternative perspective to minimize the associated costs, similar to Algorithm~\ref{alg:2}.
\noindent\begin{minipage}{0.486\textwidth}
\renewcommand\footnoterule{}     
    \removelatexerror
	\begin{algorithm}[H]\label{alg:4}
		{\footnotesize
			\caption{\mbox{Iterative resource re-adjustment procedure}}
			\SetKwFunction{Union}{Union}\SetKwFunction{FindCompress}{FindCompress}
			\SetKwInOut{Input}{input}\SetKwInOut{Output}{output}
			\Input{The network configuration $\rho_{(s_k,d_g)}\; \forall s_k \in \mathcal{O}_S, d_g\in \mathcal{O}_D$ (see Algorithm~\ref{alg:3}), maximum amount of resource re-adjustments $\Upsilon$, ${\Delta}_D$, $\delta_g$, $\forall d_g \in {\Delta}_D(t)$, load and resources of the nodes, tunable parameter $0<\epsilon<1$.}
			\Output{Re-adjusted configuration $\rho_{(s_k,d_g)} \; \forall s_k \in \mathcal{O}_S, d_g\in \mathcal{O}_D$}
			$\upsilon=0$\\
            \nonl $\backslash \backslash$ Satisfy the resource deficiencies using supply nodes with high tolerance of resource fluctuations:\\
         \While{$|{\Delta}_D|\neq 0$}{
         Select an element of ${\Delta}_D$ at random. Let $d_i$ denote that element.\\
       Find the operable supply node with the highest tolerance of resource fluctuations using~\eqref{eq:midAlg} (break ties uniformly at random). Let $k^*$ denote the index of that node.\\
         Choose a small step size $\nu$ such that $0<\nu<< min(\delta_i,C_{k^*})$\\
         $\rho_{(s_{k^*},d_i)}=\rho_{(s_{k^*},d_i)} +\nu$\\
         $\delta_i=\delta_i-\nu$\\
         $\upsilon= \upsilon+ \nu$\\
         \If{$\delta_i =0$}{
         ${\Delta}_D={\Delta}_D\setminus d_i$
         }
         }
         \nonl $\backslash \backslash$ Use the remaining capability of resource re-adjustments to increase the robustness:\\
         \While{$\upsilon\leq \Upsilon$}{
		Find the operable supply node with the highest tolerance of resource fluctuations using~\eqref{eq:midAlg} (break ties uniformly at random). Let $k'^*$ denote the index of that node.\\
		Find the operable supply node with the lowest tolerance of resource fluctuations with replacing $\argmax$ with $\argmin$ in~\eqref{eq:midAlg} (break ties uniformly at random). Let $k''^*$ denote the index of that node.\\
		$\upsilon= \upsilon+ \sum_{d_g\in \mathcal{O}_D} 2\epsilon \rho_{(s_{k''^*},d_g)}$\\	
			Re-allocate $\epsilon$ fraction of all the offered resources from supply node $k''^*$ to supply node $k'^*$: \footnote{The value of $\epsilon$ should be small enough to avoid overloading node $s_{k'^*}$.}\\
            $\rho_{(s_{k'^*},d_g)}=\rho_{(s_{k'^*},d_g)} +\epsilon \rho_{(s_{k''^*},d_g)}\;\; \forall d_g\in \mathcal{O}_D$.\\
            $\rho_{(s_{k''^*},d_g)}=(1-\epsilon)\rho_{(s_{k''^*},d_g)}\;\; \forall d_g\in \mathcal{O}_D$.
			 }
			 Output $\rho_{(s_k,d_g)} \;\forall s_k \in \mathcal{O}_S, d_g\in \mathcal{O}_D$.
		}
	\end{algorithm}	
	\end{minipage}
{\color{black}	
\begin{remark}
 The intentional failure mechanism can also be studied considering heterogeneous importance of the demand nodes. In this case, there might be a tendency toward maintaining the functionality of important demand nodes regardless of their amount of resource deficiency. 
\end{remark}
\begin{remark}
The problem of interest in this section can also be studied from the perspective of fault detection and mitigation. In that scenario, there might be intended (or unintended) misreports in the value of resources of the supply nodes and the requested loads of the demand nodes. Subsequently, the network manager's goal will be the detection of these misreports and suppressing their potential catastrophic impacts on the network. A similar problem is studied in the context of smart grid networks with a different system model and failure assumptions, e.g.,  \cite{yuan2011modeling22,6148224model,7438916Power}. However, the model proposed in these works does not explicitly focus on the interdependency between the demand and the supply. Consequently, there  is no robustness  analysis  considering  the  resource  provisioning  among  the  nodes  in these works.  Also, the corresponding effect of load fluctuations on the supply layer, resource fluctuations on the demand layer, the uniform/proportional mechanisms of resource/load fluctuations, and the cascading failure mechanism are not considered in these works. Another interesting problem is studying the network recovery after failures, e.g., \cite{wang2011progressiveee,tootaghaj3}, in the context of demand-supply networks. We leave these interesting problems to future work.
\end{remark}}

  \subsection{On the feasibility of mitigating cascading failures in demand-supply networks}  
  In this subsection, we provide some insights on the feasibility of the above-described methodology in some special scenarios with respect to the network operator's capability and the network setting.
  \\
  
  \textbf{A) Capability of the network operator:}
  In general, the capability of the network operator belongs to one of the following cases: i) performing both resource re-adjustment and intentional failures without any restrictions; ii) performing unlimited intentional failures and limited resource re-adjustments; iii) performing unlimited resource re-adjustments and limited intentional failures; iv) performing limited resource re-adjustments and intentional failures. In the following, we provide a discussion on the  performance of Algorithms~\ref{alg:3},~\ref{alg:4} for each case.  
  At any stage of cascading failure, let $\Delta_D=\{\delta_1,\delta_2,\cdots,\delta_{|\Delta_D|}\}$ denote the set of resource deficiencies of those demand nodes having resource deficiency and assume that $\delta_1\geq \delta_2\geq ...\geq \delta_{|\Delta_D|}>0$.  In the first case, it is trivial to verify that no cascading failure will be spread upon using our proposed algorithms. 
  In the second case, it can be verified that the cascading failure will not spread since all the demand nodes with resource deficiency can be isolated from the supply layer. However, failing of all the demand nodes is obviously not the best strategy.
  In this case, the capability of the network operator in performing resource readjustments $\Upsilon$ identifies the minimum required number of failed demand nodes $f^*_d$ to stop an ongoing cascading failure. Mathematically, this parameter can be obtained as:
  \vspace{1mm}
  \begin{equation}\label{eq:con1}
     \Scale[.90]{ f^*_d=\displaystyle\argmin_{f_d\in \mathbb{Z}^+} \sum_{i=1}^{|\Delta_D|} \delta_i- \sum_{i=1}^{f_d} \delta_i\leq \Upsilon \equiv \displaystyle\argmin_{f_d\in \mathbb{Z}^+} \sum_{i=f_d+1}^{|\Delta_D|} \delta_i\leq \Upsilon.}
      \vspace{1.5mm}
  \end{equation}
  In contrast, in the third case, the capability of the network operator in performing intentional failures $\Gamma$ identifies the minimum required amount of resource readjustments. Mathematically, in this case, at any stage of the cascading failure, if  the following constraint is met, then the cascading failure can be absorbed by performing resource re-adjustments: 
  \begin{equation}\label{eq:con2}
      \sum_{i=1}^{S} {C_i}\geq \sum_{i=\Gamma+1}^{|\Delta_D|} \delta_i.
  \end{equation}
  In the fourth case, at any stage of cascading failures, the possibility of mitigation should be checked by assuming failure of the maximum number of demand nodes with resource deficiency using~\eqref{eq:con2}. If~\eqref{eq:con2} holds, then the minimum number of required intentional failures can be found based on~\eqref{eq:con1}; otherwise, the algorithm waits for the next time instant for deployment of the adaptability schemes. Note that our proposed algorithms enjoy low computational complexities. In particular, the computational complexity of performing $\Gamma$ intentional failures in Algorithm~\ref{alg:3} is $O(\Gamma)$, while the computational complexity of performing $\Upsilon$ resource re-adjustments in Algorithm~\ref{alg:4} is $O(\frac{\sum_{i=1}^{|\Delta_D|}\delta_i}{\nu}+\frac{\Upsilon-\sum_{i=1}^{|\Delta_D|}\delta_i}{\epsilon})$, where $\nu$ and $\epsilon$ are tuning parameters used in Algorithm~\ref{alg:4}.
  The aforementioned constraints (\eqref{eq:con1},~\eqref{eq:con2}) also provide interesting insights for an attacker, especially in scenarios such as battlefields, who aims to attack some supply nodes or impose targeted resource deficiencies on some demand node by restraining the resource provisioning to them. In particular, they can identify the required amount of resource deficiencies on the demand layer to trigger cascading failures.
  \\
  
    \textbf{B) Networks with a stubborn/fixed topology:} Our proposed resource re-adjustment method in Algorithm~\ref{alg:4} assumes the possibility of adding a link between a pair of nodes, which were previously sharing no resource, i.e., changing the weight of an edge with the zero initial weight. For a network with a fixed topology, or equivalently when the construction of a link takes forbidding efforts, this assumption may not be valid. In this case, the resource re-adjustments are limited to modifying the resources among the initially connected nodes, i.e., changing the weights of those edges with non-zero weights. For this scenario, we first provide a necessary condition for mitigating cascading failures upon having a limited capability of performing intentional failures and then provide the corresponding resource re-adjustment scenario.
    
Since the network corresponds to a bipartite graph, the adjacency matrix of the network $A$ can be written as follows:
\begin{equation}
    A= \begin{bmatrix}
   0& B\\
   B^T&0
   \end{bmatrix}.
\end{equation}
The matrix A has $D+S$ rows and columns, where row $1$ to $S$ corresponds to supply node $s_1$ to $s_S$ and rows $S+1$ to $S+D$ correspond to demand node $d_1$ to $d_D$; likewise for the columns. The rectangular matrix $B=[b_{ij}]_{1\leq i\leq S,\; 1\leq j\leq D}$ identifies the existence of a connection between a pair nodes, where $b_{ij}=1$ if $\rho(s_i,d_j)>0$, and zero otherwise. In this case, upon existence of any resource deficiencies on the demand layer, let $\tilde{r}_{ij}$ denote the extra resources needed to be allocated from node $s_i$ to demand node $d_j$ to satisfy the stability condition of that demand node. 
We group the supply nodes as $\mathcal{N}(d_1), \mathcal{N}(d_2), \cdots, \mathcal{N}(d_D)$, where $\mathcal{N}(d_j)$ denotes a group consisting of those supply nodes providing resources to demand node $d_j$. A supply node may belong to multiple groups upon providing to multiple demand nodes. Consider the set $\Delta_D$ with its elements as defined above and let $d_{k_1}, d_{k_2},\cdots, d_{k_{|\Delta_D|}}$ denote those demand nodes with resource deficiencies and $P\triangleq \mathcal{N}(d_{k_1})\cup \mathcal{N}(d_{k_2})\cdots \cup \mathcal{N}(d_{k_{|\Delta_D|}})$. It can be verified that the necessary condition for mitigating an ongoing cascading failure is then given by $\sum_{s_i\in P} {C_i}\geq \sum_{i=\Gamma+1}^{|\Delta_D|} \delta_i$.

In this scenario, a suitable resource re-adjustment $\tilde{R}=[\tilde{r}_{ij}]_{1\leq i\leq S,\; 1\leq j\leq D}$ should satisfy the following set of equations: 
\begin{equation}\label{eq:numerical}
\begin{aligned}
   &b_{11} \tilde{r}_{1,1}+b_{21} \tilde{r}_{2,1}+b_{31} \tilde{r}_{3,1}+\cdots +b_{S1} \tilde{r}_{S,1}=\delta_1,\\
    &  b_{12} \tilde{r}_{1,2}+b_{22} \tilde{r}_{2,2}+b_{32} \tilde{r}_{3,2}+\cdots +b_{S2} \tilde{r}_{S,2}=\delta_2,\\
   &\vdots\\
   & b_{1D} \tilde{r}_{1,D}+b_{2D} \tilde{r}_{2,D}+b_{3D} \tilde{r}_{3,D}+\cdots +b_{SD} \tilde{r}_{S,D}=\delta_D.
 \end{aligned}
\end{equation}
The left hand side of this system of equations can be written as the Hadamard product of matrices $B$ and $\tilde{R}$. In Algebra, there is no known solution for the above equation holding for any matrix $B$. Hence, the set of solutions of~\eqref{eq:numerical} must be found numerically. In that case, among the set of solutions, any solution $\tilde{{R}}^*=[\tilde{r}^*_{ij}]_{1\leq i\leq S,\; 1\leq j\leq D}$ satisfying $ \tilde{r}^*_{i,1}+ \tilde{r}^*_{i,2}+ \tilde{r}^*_{i,3}+\cdots+ \tilde{r}^*_{i,D}\leq R_i$, $\forall s_i \in \mathcal{S}$, satisfies the stability conditions on the demand and supply nodes, and thus terminates the propagation of cascading failures. It is obvious that upon applying intentional failures, some elements on the right hand side of~\eqref{eq:numerical} become zero.

	
	 \begin{figure}[t]
	\minipage{8.7cm}
		\includegraphics[width=1\linewidth, height=0.649\linewidth]{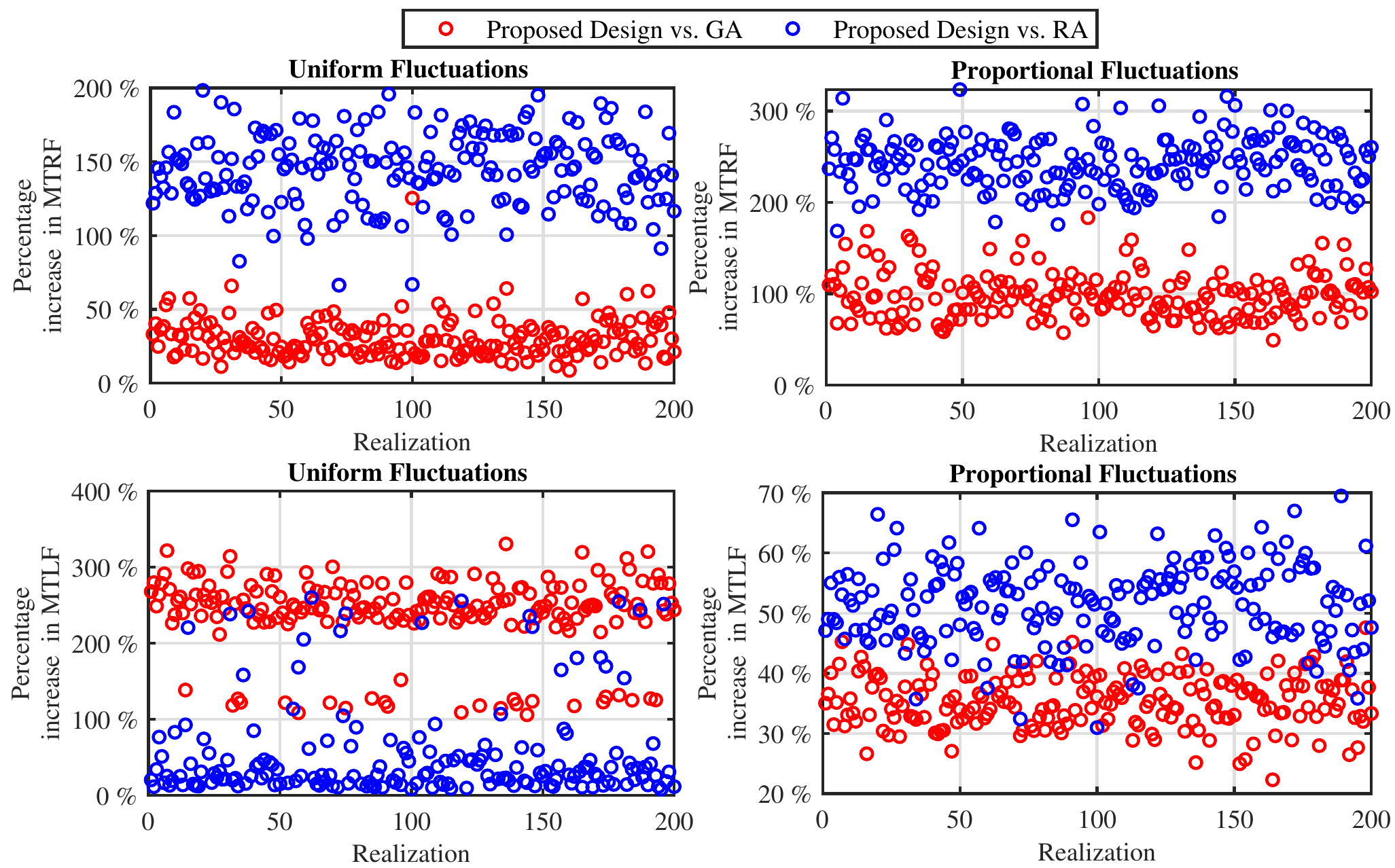}
		\caption{Percentage increase in $MTRF$ and $MTLF$ using our proposed design methods as compared to the baselines under uniform fluctuations (the subplots on the left) and proportional fluctuations (the subplots on the right). \label{fig:11}}
		\endminipage
		\vspace{4mm}
	\minipage{8.7cm}
	\includegraphics[width=1\linewidth, height=0.649\linewidth]{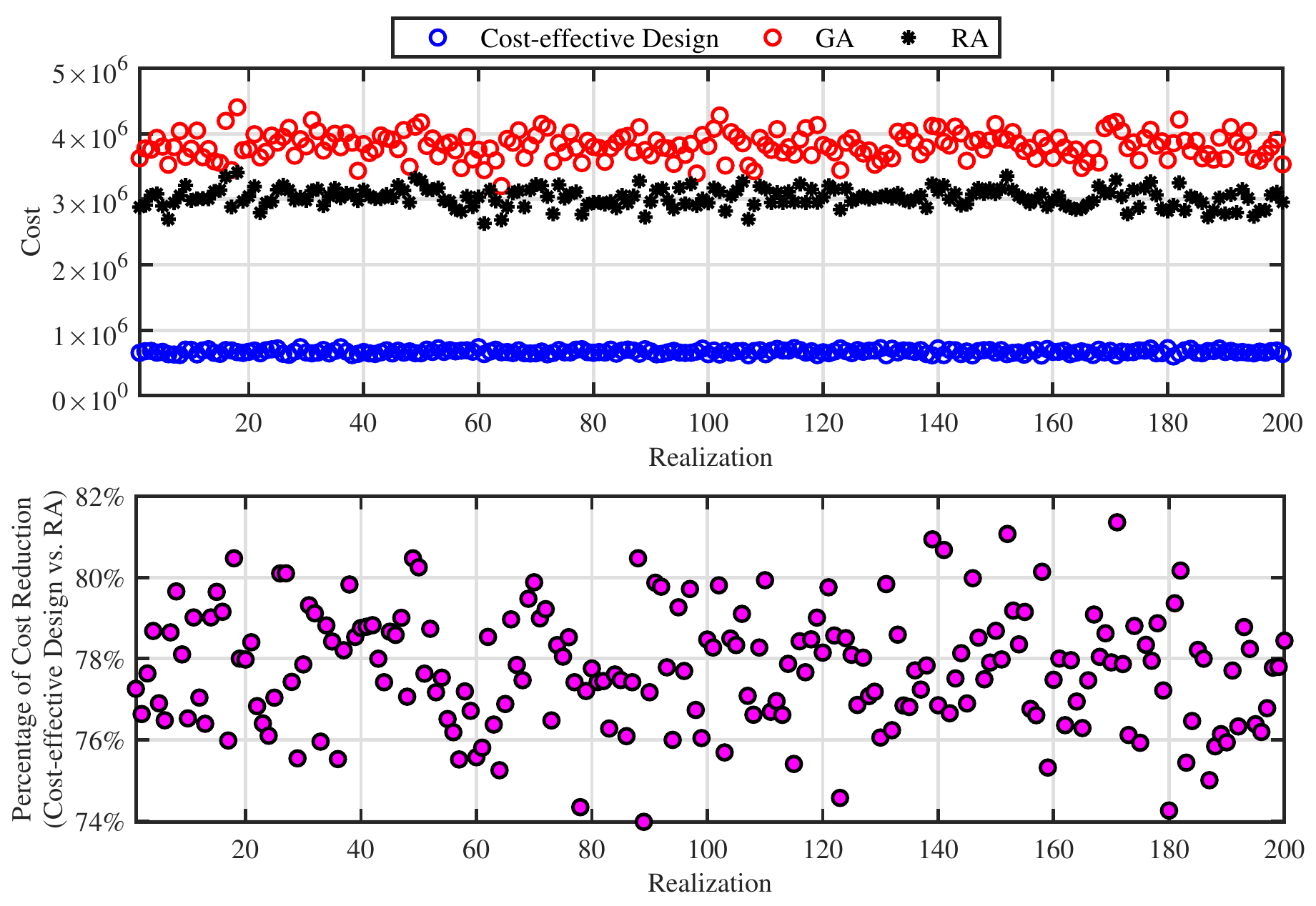}
		\caption{Comparison between allocation cost of different resource allocation methods (top). Percentage of decrease in cost upon using our cost-effective design as compared to RA (bottom).\label{fig:22}}
		\endminipage
			\vspace{5mm}
	\minipage{8.7cm}
		\includegraphics[width=1\linewidth, height=0.649\linewidth]{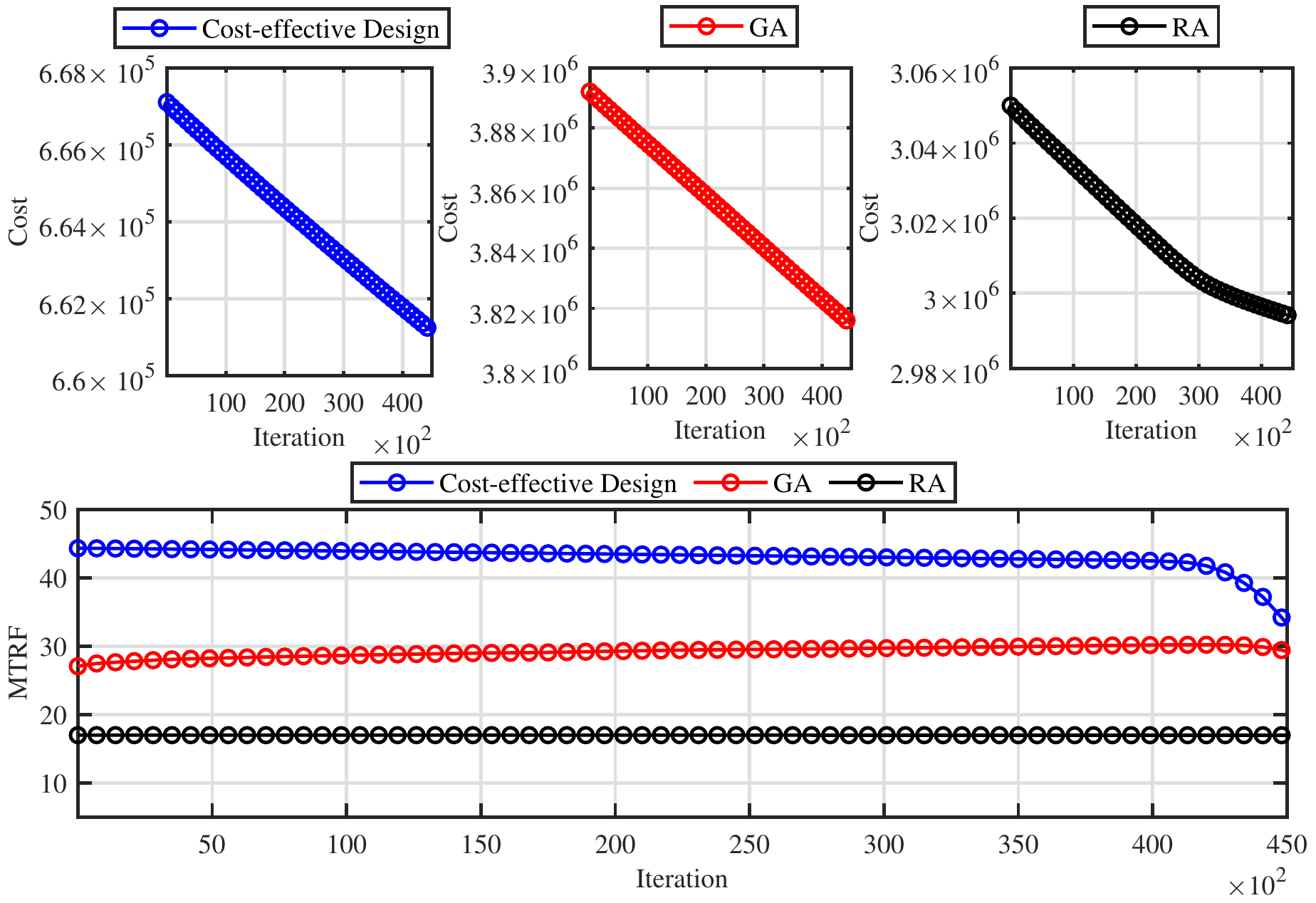}
		\caption{Cost of allocation versus iteration for Algorithm~\ref{alg:2} (top). Comparison between $MTRF$ of different allocation methods upon cost reduction under uniform resource fluctuations (bottom). \label{fig:33}}
		\endminipage
	\end{figure}

  \begin{figure}  
    	\includegraphics[width=0.99\linewidth, height=0.7\linewidth]{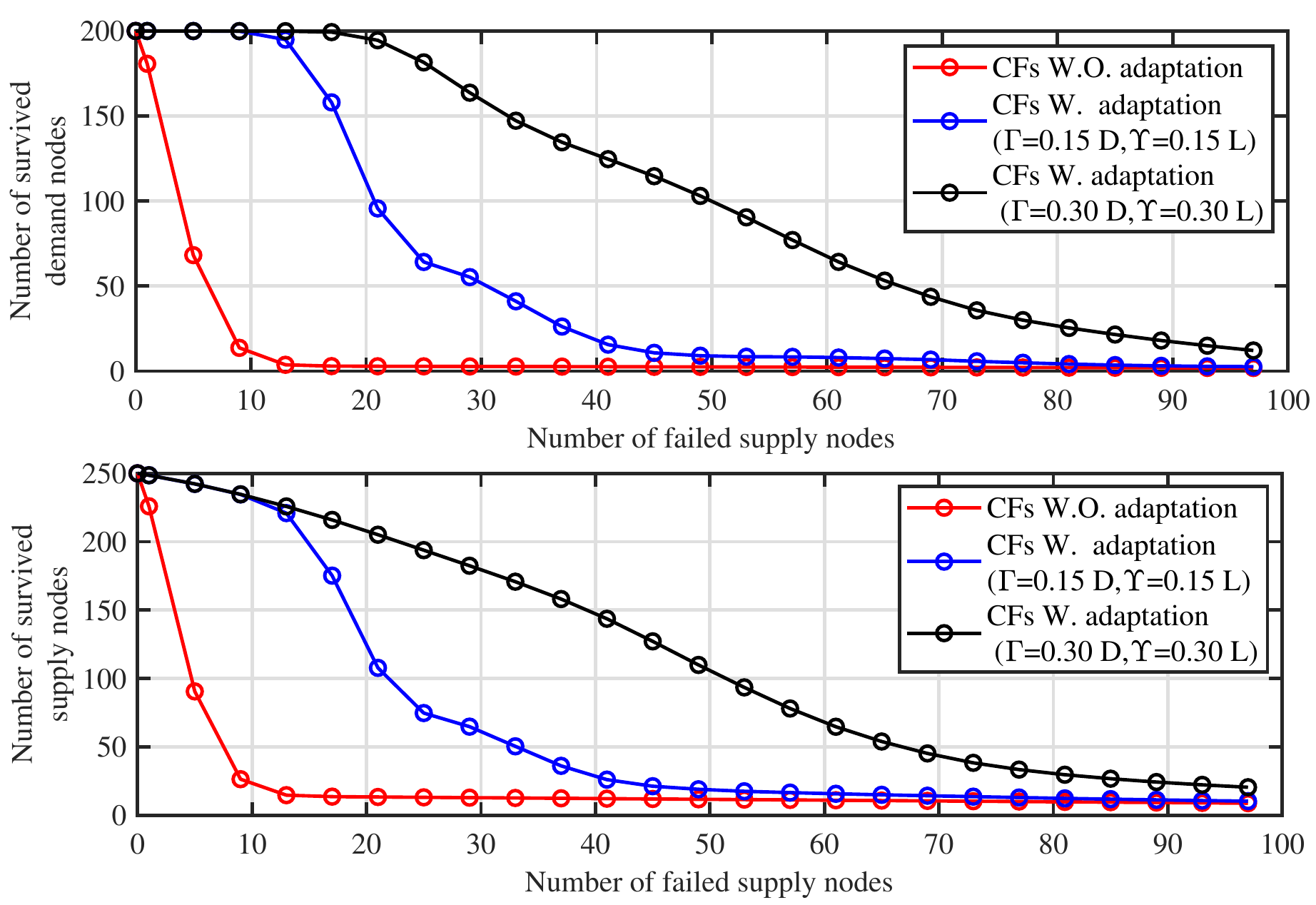}
		\caption{Comparison between the number of surviving supply/demand nodes under cascading failures with and without utilizing our network adaptation mechanism. In the legend, abbreviations CFs., W.O and W. denote cascading failures, without, and with, respectively, $D$ indicates the number of demand nodes and $L=\sum_{d_i\in \mathcal{D}} L_i$. \label{fig:44}}
    \end{figure}
    
      \begin{figure}  
    	\includegraphics[width=0.99\linewidth, height=0.7\linewidth]{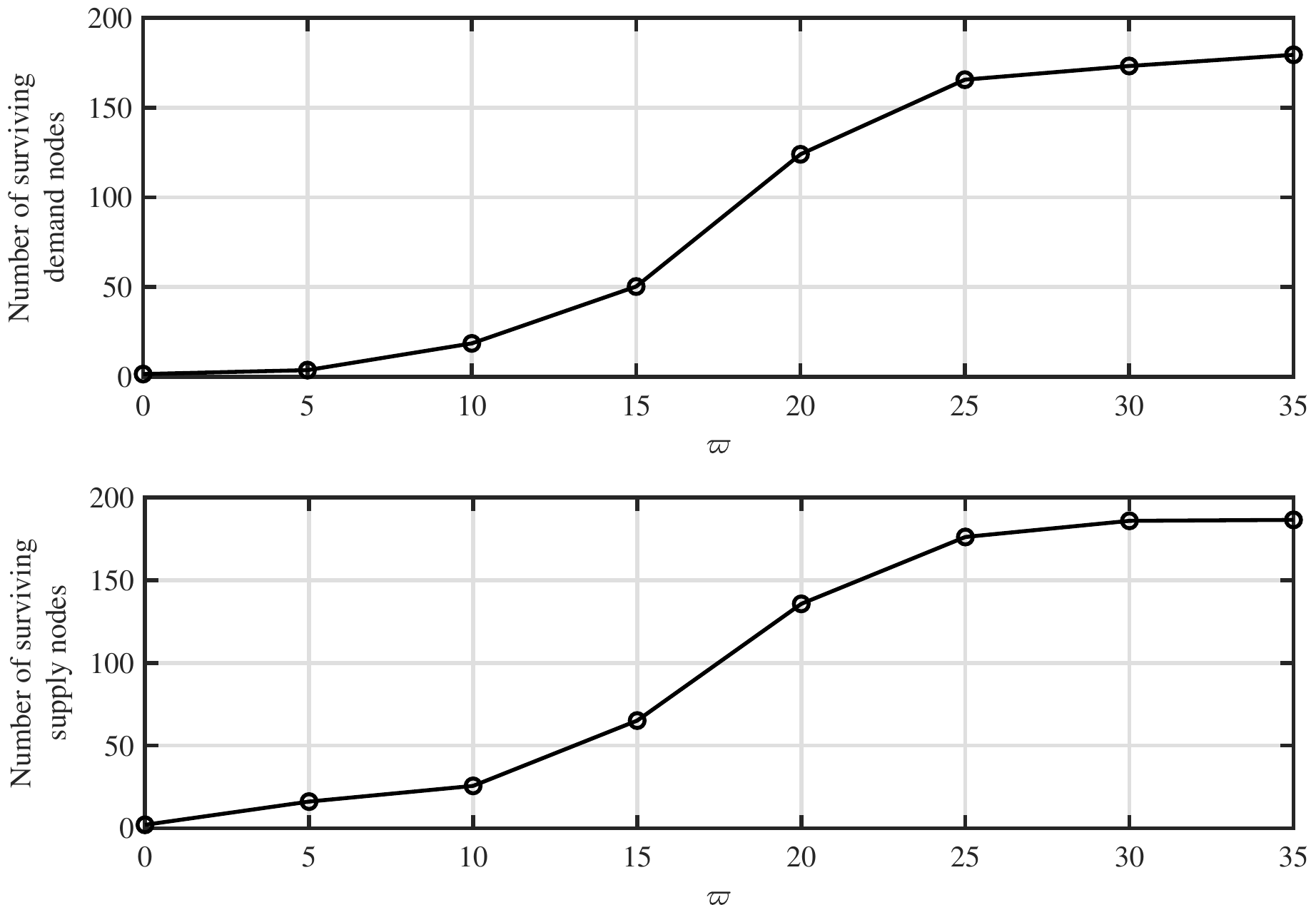}
		\caption{The number of surviving supply/demand nodes under cascading failures considering different network adaptation capabilities described by the parameter $\varpi$, where $\Gamma=\varpi D$, $\Upsilon= \varpi L$, and $L=\sum_{d_i\in \mathcal{D}} L_i$. \label{fig:45}}
		\vspace{1mm}
    \end{figure}
    
    	\section{Simulation Results}\label{sec:Sims}
\noindent We consider a demand-supply network consisting of $250$ supply nodes and $200$ demand nodes. Simulations are conducted for $200$ realizations of the network, in each of which the resources and loads of the respective nodes are generated uniformly at random in the intervals $[10,280]$ and $[10,200]$, respectively.

  We consider two allocation strategies, greedy allocation (GA) and random allocation (RA), as baselines. The GA is an iterative algorithm, which is fed with the resource pool and the requested loads of the nodes.
At each iteration, the GA chooses the supply node with the largest spare resource and fulfills the demands of the demand node with the largest unsatisfied load. In the RA, the adjacency matrix of the network is generated randomly fulfilling the stability conditions. In both methods, before realizing the resource configuration, we freeze/obstruct $10\%$ of the resources of each supply node to avoid having supply nodes at the threshold of failure. In the cost modeling, i.e.,~\eqref{eq:cost}, for each pair of nodes, parameter $\alpha$ is chosen uniformly at random between $10$ and $100$, while $\beta$ is chosen uniformly at random between $1.1$ and $1.4$ . In Algorithm~\ref{alg:2}, we set $\epsilon=0.01$. 

Fig.~\ref{fig:11} depicts the percentage increase in the values of $MTRF$ and $MTLF$ obtained by utilizing our proposed design methods as compared to the baselines under uniform and proportional resource/load fluctuations. From the two subplots on the top, it can be seen that on average our method exhibits around $87\%$ ($175\%$) performance gain in MTRF as compared to the baselines\footnote{These numbers are the average performance gains over the two baselines.} upon uniform (proportional) resource fluctuations. Also, from the two subplots on the bottom, $165\%$ ($42\%$) performance gain in MTLF as compared to the baselines upon uniform (proportional) load fluctuations can be observed. We have observed that our proposed methods lead to significant performance gains as compared to the baselines for other network settings. However, the amount of achieved gains may vary form one parameter setting to another for different metrics and load fluctuation mechanisms. As can be seen from Fig.~\ref{fig:11}, the greedy algorithm often outperforms the random allocation, which is manifested by a closer performance to our optimal designs, except for MTLF with the uniform load fluctuations (bottom left subplot). This may be explained as follows: in the greedy algorithm most of the demand nodes are supplied by a few number of supply nodes. Hence, the amount of load fluctuation on the demand nodes does not spread among multiple supply nodes, which directly endangers those supply nodes with a small free capacity. 

We depict the allocation cost of different resource configuration methods in Fig.~\ref{fig:22}. From the top plot of this figure, the clear superiority of our cost-effective design can be seen. Also, as can be seen, the RA results in a less allocation cost as compared to the GA. This is due to the resource congestion over the links connected to the supply nodes with large value of resources in the GA. To have a better comparison, the percent of decrease in cost upon using our cost-effective design as compared to the RA is shown in the bottom plot of this figure, which reveals a $77\%$ cost reduction (on average) upon using our method.

For the uniform resource fluctuations scenario, Fig.~\ref{fig:33} demonstrates the average performance of our cost-reduction method (Algorithm~\ref{alg:2}). In the top plots, the cost is shown versus iteration demonstrating the cost reduction upon applying our algorithm to different resource allocation methods. The corresponding $MTRF$ of the network is depicted in the bottom plot. As can be seen, our algorithm successfully reduces the allocation cost while maintaining a high robustness for the network. Also, as can be seen from the bottom plot, at the end of the last iteration although $MTRF$s for our cost-effective design and the baselines, especially GA, are close together, the cost of the cost-effective design is significantly lower.

The performance of confining cascading failures using our proposed network adaptability schemes is depicted in Fig.~\ref{fig:44}. In this simulation, we measure the number of surviving/stable nodes at the end of cascading failures triggered by failing various numbers of supply nodes. To have a fair performance measurement, we apply our algorithm to $200$ realizations of the network configured using RA. Each point in the figure represents the average performance over $20000$ iterations comprising $200$ network realizations with $100$ different choices for the initially failed supply nodes for each realization. The simulations are conducted for two choices of the parameters for $\Gamma$ and $\Upsilon$ (introduced in Section~\ref{sec:confine}): i) $\Gamma=0.15 D$, $\Upsilon= 0.15 \sum_{d_i\in \mathcal{D}} L_i$, and ii) $\Gamma=0.35 D$, $\Upsilon= 0.35 \sum_{d_i\in \mathcal{D}} L_i$, which corresponds to a higher capability of performing demand nodes isolation and load re-adjustments as compared to the first case. {\color{black} As can be seen, implementing our network adaptation scheme can have a remarkable impact on the size of the surviving network, and thus the survivability of the network upon occurrence of cascading failures.}
{\color{black} Also, the impact of different network adaptation capabilities on the final size of the network upon occurrence of cascading failures is depicted in Fig.~\ref{fig:45}. In this simulation, initially, $30$ supply nodes are randomly failed and the results are obtained through averaging over $100$ Monte-Carlo iterations. It can be seen that as the capability of adapting the network increases, our proposed algorithm can save a significant number of nodes from failures. 
}
\section{Conclusion and Future Work}\label{sec:conc}
\noindent In this work, we studied the robustness of demand-supply networks by considering the resource as a quantitative value and incorporating the inherent resource sharing mechanism into our model. We studied the effect of different stress mechanisms on the network and investigated a suitable cascading failure mechanism. After quantifying the robustness considering different load/resource fluctuation scenarios, we proposed effective methods achieving the highest robustness to prevent cascading failures from happening. We further proposed a method that achieves the highest robustness under heterogeneous resource allocation costs among the nodes. We introduced an effective algorithm to reduce the resource allocation cost while maintaining a high robustness. Moreover, we extended the concept of network adaptability to our generic model. Along this direction, we proposed new network adaptability methods, using which we developed an algorithm to confine cascading failures. For the future work, one interesting direction is to study  the demand-supply networks with uncertainties in available resources and requested loads of the nodes. This requires utilizing probabilistic methods to quantify the robustness and solve the optimization problems.   

  \section*{Acknowledgments}
This work was supported in part by the National Science Foundation under grants ECCS-1444009 and CNS-1824518, and in part by Army Research Office under Grant W911NF-17-1-0087. 

   \bibliographystyle{IEEEtran}
\bibliography{IEEEabrv}

\begin{thebibliography}{10}
\providecommand{\url}[1]{#1}
\csname url@samestyle\endcsname
\providecommand{\newblock}{\relax}
\providecommand{\bibinfo}[2]{#2}
\providecommand{\BIBentrySTDinterwordspacing}{\spaceskip=0pt\relax}
\providecommand{\BIBentryALTinterwordstretchfactor}{4}
\providecommand{\BIBentryALTinterwordspacing}{\spaceskip=\fontdimen2\font plus
\BIBentryALTinterwordstretchfactor\fontdimen3\font minus
  \fontdimen4\font\relax}
\providecommand{\BIBforeignlanguage}[2]{{%
\expandafter\ifx\csname l@#1\endcsname\relax
\typeout{** WARNING: IEEEtran.bst: No hyphenation pattern has been}%
\typeout{** loaded for the language `#1'. Using the pattern for}%
\typeout{** the default language instead.}%
\else
\language=\csname l@#1\endcsname
\fi
#2}}
\providecommand{\BIBdecl}{\relax}
\BIBdecl

\bibitem{CatastrophicCascadeofFailuresInInterdependentNetworks}
S.~V. Buldyrev, R.~Parshani, G.~Paul, H.~E. Stanley, and S.~Havlin,
  ``Catastrophic cascade of failures in interdependent networks,''
  \emph{Nature}, vol. 464, no. 7291, p. 1025, 2010.

\bibitem{TheRobustnessofInterdependentTransportationNetworksUnderTargetedAttack}
P.~Zhang, B.~Cheng, Z.~Zhao, D.~Li, G.~Lu, Y.~Wang, and J.~Xiao, ``The
  robustness of interdependent transportation networks under targeted attack,''
  \emph{EPL (Europhysics Lett.)}, vol. 103, no.~6, p. 68005, 2013.

\bibitem{CascadingFailuresinInterdependentNetworkswithMultipleSupply-DemandLinksandFunctionalityThresholds}
M.~Di~Muro, L.~Valdez, H.~A. R{\^e}go, S.~Buldyrev, H.~Stanley, and
  L.~Braunstein, ``Cascading failures in interdependent networks with multiple
  supply-demand links and functionality thresholds,'' \emph{Scientific Rep.},
  vol.~7, no.~1, p. 15059, 2017.

\bibitem{arxiveConnectivity}
J.~Zhang and E.~Modiano, ``Connectivity in interdependent networks,''
  \emph{IEEE/ACM Trans. Netw.}, vol.~26, no.~5, pp. 2090--2103, Oct. 2018.

\bibitem{Inference}
F.~Poppe, J.~Jones, S.~Venkatachalam, S.~Dharanikota, R.~Jain, R.~Hartani,
  D.~Griffith, and Y.~Xue, ``Inference of shared risk link groups,''
  \emph{Internet Draft}, 2001.

\bibitem{Diverse}
J.~Q. Hu, ``Diverse routing in optical mesh networks,'' \emph{IEEE Trans.
  Commun.}, vol.~51, no.~3, pp. 489--494, 2003.

\bibitem{approximability}
D.~Coudert, P.~Datta, S.~P{\'e}rennes, H.~Rivano, and M.-E. Voge, ``Shared risk
  resource group complexity and approximability issues,'' \emph{Parallel
  Process. Lett.}, vol.~17, no.~02, pp. 169--184, 2007.

\bibitem{RecoveryOfInterdependentNetworks}
M.~Di~Muro, C.~La~Rocca, H.~Stanley, S.~Havlin, and L.~Braunstein, ``Recovery
  of interdependent networks,'' \emph{Scientific Rep.}, vol.~6, p. 22834, 2016.

\bibitem{StrategyForStoppingFailureCascadesInInterdependentNetworks}
C.~E.~L. Rocca, H.~E. Stanley, and L.~A. Braunstein, ``Strategy for stopping
  failure cascades in interdependent networks,'' \emph{Physica A: Statistical
  Mech. Appl.}, vol. 508, pp. 577 -- 583, 2018.

\bibitem{NetworkRobustnessofMultiplexNetworkswithInterlayerDegreeCorrelations}
B.~Min, S.~Do~Yi, K.-M. Lee, and K.-I. Goh, ``Network robustness of multiplex
  networks with interlayer degree correlations,'' \emph{Physical Rev. E},
  vol.~89, no.~4, p. 042811, 2014.

\bibitem{DesigningOptimalInterlinkPatternstoMaximizeRobustnessofInterdependentNetwork}
S.~Chattopadhyay, H.~Dai, D.~Y. Eun, and S.~Hosseinalipour, ``Designing optimal
  interlink patterns to maximize robustness of interdependent networks against
  cascading failures,'' \emph{IEEE Trans. Commun.}, vol.~65, no.~9, pp.
  3847--3862, Sept 2017.

\bibitem{ImprovingRobustnessofInterdependentNetworksbyaNewCouplingStrategy}
X.~Wang, W.~Zhou, R.~Li, J.~Cao, and X.~Lin, ``Improving robustness of
  interdependent networks by a new coupling strategy,'' \emph{Physica A:
  Statistical Mech. Appl.}, vol. 492, pp. 1075 -- 1080, 2018.

\bibitem{CyberPhysical}
Z.~Huang, C.~Wang, M.~Stojmenovic, and A.~Nayak, ``Characterization of
  cascading failures in interdependent cyber-physical systems,'' \emph{IEEE
  Trans. Comput.}, vol.~64, no.~8, pp. 2158--2168, Aug 2015.

\bibitem{markovChainCascades}
M.~Rahnamay-Naeini and M.~M. Hayat, ``Cascading failures in interdependent
  infrastructures: An interdependent {M}arkov-chain approach,'' \emph{IEEE
  Trans. Smart Grid}, vol.~7, no.~4, pp. 1997--2006, July 2016.

\bibitem{huang2011robustness}
X.~Huang, J.~Gao, S.~V. Buldyrev, S.~Havlin, and H.~E. Stanley, ``Robustness of
  interdependent networks under targeted attack,'' \emph{Physical Rev. E},
  vol.~83, no.~6, p. 065101, 2011.

\bibitem{ResilienceD2D1}
S.~A. Pambudi, W.~Wang, and C.~Wang, ``On the resilience of {D2D}-based social
  networking service against random failures,'' in \emph{proc. IEEE Global
  Commun. Conf. (GLOBECOM)}, Dec 2016, pp. 1--6.

\bibitem{ResilienceD2D2}
------, ``From isolation time to node resilience: Impact of cascades in
  {D2D}-based social networks,'' in \emph{proc. IEEE Global Commun. Conf.
  (GLOBECOM)}, Dec 2017, pp. 1--6.

\bibitem{tootaghaj1}
D.~Z. Tootaghaj, N.~Bartolini, H.~Khamfroush, and T.~La~Porta, ``Controlling
  cascading failures in interdependent networks under incomplete knowledge,''
  in \emph{Proc. IEEE 36th Symp. Rel. Distrib. Syst. (SRDS)}, 2017, pp. 54--63.

\bibitem{tootaghaj2}
D.~Z. Tootaghaj, N.~Bartolini, H.~Khamfroush, T.~He, N.~R. Chaudhuri, and
  T.~La~Porta, ``Mitigation and recovery from cascading failures in
  interdependent networks under uncertainty,'' \emph{IEEE Trans. Control Netw.
  Syst.}, 2018.

\bibitem{RobustnessofInterdependentNetworksWithDifferentLinkPatternsAgainstCascadingFailures}
J.~Wang, C.~Jiang, and J.~Qian, ``Robustness of interdependent networks with
  different link patterns against cascading failures,'' \emph{Physica A:
  Statistical Mech. Appl.}, vol. 393, pp. 535 -- 541, 2014.

\bibitem{ERfailure}
A.~Eslami, C.~Huang, J.~Zhang, and S.~Cui, ``Cascading failures in
  load-dependent finite-size random geometric networks,'' \emph{IEEE Trans.
  Netw. Sci. Eng.}, vol.~3, no.~4, pp. 183--196, Oct 2016.

\bibitem{impactOfTopology}
P.~Dey, R.~Mehra, F.~Kazi, S.~Wagh, and N.~M. Singh, ``Impact of topology on
  the propagation of cascading failure in power grid,'' \emph{IEEE Trans. Smart
  Grid}, vol.~7, no.~4, pp. 1970--1978, July 2016.

\bibitem{NetworkAdaptabilityFromDisasterDisruptionsAndCascadingFailures}
B.~Mukherjee, M.~F. Habib, and F.~Dikbiyik, ``Network adaptability from
  disaster disruptions and cascading failures,'' \emph{IEEE Commun. Mag.},
  vol.~52, no.~5, pp. 230--238, May 2014.

\bibitem{ImproveNetworksRobustnessAgainstCascadewithRewiring}
H.~A.~Q. Tran and A.~Namatame, ``Improve network's robustness against cascade
  with rewiring,'' \emph{Procedia Comput. Sci.}, vol.~24, pp. 239--248, 2013.

\bibitem{ApproximationAndHardnessResultsForLabelCutandRelatedProblems}
P.~Zhang, J.-Y. Cai, L.-Q. Tang, and W.-B. Zhao, ``Approximation and hardness
  results for label cut and related problems,'' \emph{J. Combinatorial
  Optimization}, vol.~21, no.~2, pp. 192--208, Feb 2011.

\bibitem{ComplexityAndApproximabilityIssuesOfSharedRiskResourceGroup}
D.~Coudert, P.~Datta, S.~P{\'e}rennes, H.~Rivano, and M.-E. Voge, ``Complexity
  and approximability issues of shared risk resource group,'' Ph.D.
  dissertation, INRIA, 2006.

\bibitem{CascadingFailuresInbi-partiteGraphsModelForSystemicRiskPropagation}
X.~Huang, I.~Vodenska, S.~Havlin, and H.~E. Stanley, ``Cascading failures in
  bi-partite graphs: model for systemic risk propagation,'' \emph{Scientific
  Rep.}, vol.~3, p. 1219, 2013.

\bibitem{StabilityAnalysisofFinancialContagionDuetoOverlappingPortfolios}
F.~Caccioli, M.~Shrestha, C.~Moore, and J.~D. Farmer, ``Stability analysis of
  financial contagion due to overlapping portfolios,'' \emph{J. Banking \&
  Finance}, vol.~46, pp. 233--245, 2014.

\bibitem{UniformLines1}
T.~C. Gulcu, V.~Chatziafratis, Y.~Zhang, and O.~Yagan, ``Attack vulnerability
  of power systems under an equal load redistribution model,'' \emph{IEEE/ACM
  Trans. Netw.}, vol.~26, no.~3, pp. 1306--1319, Jun. 2018.

\bibitem{UniformLines2}
O.~Ya{\u{g}}an, ``Robustness of power systems under a
  democratic-fiber-bundle-like model,'' \emph{Physical Rev. E}, vol.~91, no.~6,
  p. 062811, 2015.

\bibitem{UniformLines3}
Y.~Zhang and O.~Ya{\u{g}}an, ``Optimizing the robustness of electrical power
  systems against cascading failures,'' \emph{Scientific Rep.}, vol.~6, p.
  27625, 2016.

\bibitem{ConvexOptimization}
S.~Boyd and L.~Vandenberghe, \emph{Convex optimization}.\hskip 1em plus 0.5em
  minus 0.4em\relax Cambridge university press, 2004.

\bibitem{yuan2011modeling22}
Y.~Yuan, Z.~Li, and K.~Ren, ``Modeling load redistribution attacks in power
  systems,'' \emph{IEEE Trans. Smart Grid}, vol.~2, no.~2, pp. 382--390, 2011.

\bibitem{6148224model}
Y.~{Yuan}, Z.~{Li}, and K.~{Ren}, ``Quantitative analysis of load
  redistribution attacks in power systems,'' \emph{IEEE Trans. Parallel
  Distrib. Syst.}, vol.~23, no.~9, pp. 1731--1738, Sep. 2012.

\bibitem{7438916Power}
G.~{Liang}, J.~{Zhao}, F.~{Luo}, S.~R. {Weller}, and Z.~Y. {Dong}, ``A review
  of false data injection attacks against modern power systems,'' \emph{IEEE
  Trans. Smart Grid}, vol.~8, no.~4, pp. 1630--1638, July 2017.

\bibitem{wang2011progressiveee}
J.~Wang, C.~Qiao, and H.~Yu, ``On progressive network recovery after a major
  disruption,'' in \emph{Proc. IEEE Int. Conf. Comp. Commun. (INFOCOM)}, 2011,
  pp. 1925--1933.

\bibitem{tootaghaj3}
D.~{Z. Tootaghaj}, N.~{Bartolini}, H.~{Khamfroush}, and T.~{La Porta}, ``On
  progressive network recovery from massive failures under uncertainty,''
  \emph{IEEE Trans. Netw. Serv. Manag.}, vol.~16, no.~1, pp. 113--126, March
  2019.

\end{thebibliography}

\begin{IEEEbiography}[{\includegraphics[width=1.0in,height=1.15in,clip]{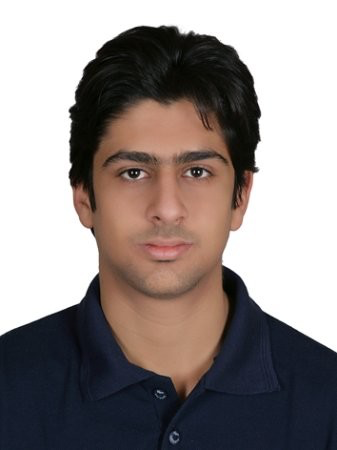}}]{Seyyedali Hosseinalipour} (S'17)  received the B.S. degree in Electrical Engineering from Amirkabir University
of Technology (Tehran Polytechnic), Tehran, Iran in 2015. He is currently pursuing a Ph.D. degree in
the Department of Electrical and Computer Engineering at North Carolina State University, Raleigh, NC, USA. His research interests include analysis of wireless networks, resource allocation and load balancing for cloud computing, and studying the robustness of interdependent networks.
\end{IEEEbiography}
\vspace{-11mm}
\begin{IEEEbiography}[{\vspace{-5mm}\includegraphics[width=1.02in,height=1.05in,clip]{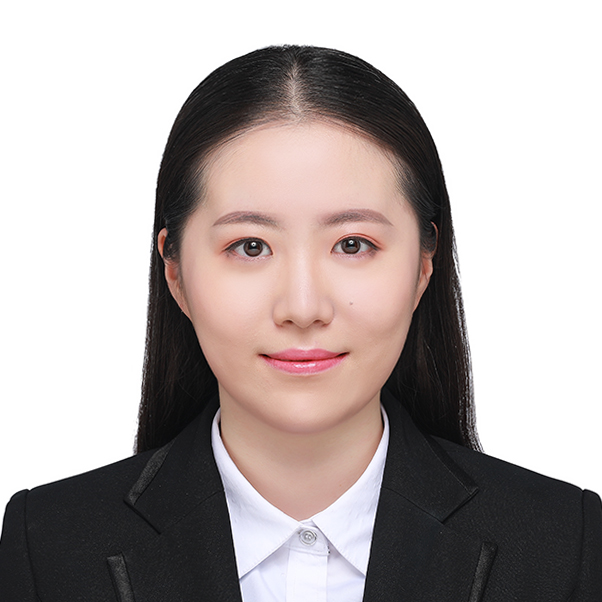}}]{Jiayu Mao} received the B.S. degree in Electronic Science \& Technology from Beijing Institute of Technology, China, in 2019. She is currently pursuing a PhD degree in the Department of Electrical Engineering, Penn State University. Her current research interests include machine learning, edge learning and wireless communication.
\end{IEEEbiography}
\vspace{-12mm}
\begin{IEEEbiography}[{\includegraphics[width=1.12in,height=1.25in,clip]{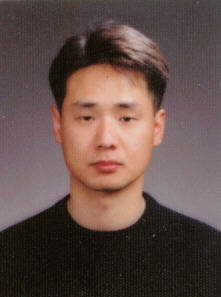}}]{Do Young Eun}   received his B.S. and M.S. degree in Electrical Engineering from Korea Advanced Institute of Science and Technology (KAIST), Taejon, Korea, in 1995 and 1997, respectively, and Ph.D. degree from Purdue University, West Lafayette, IN, in 2003. Since August 2003, he has been with the Department of Electrical and Computer Engineering at North Carolina State University, Raleigh, NC, where he is now a professor. His research interests include network modeling and performance analysis, mobile ad-hoc/sensor networks, mobility modeling, and randomized algorithms for large (social) networks. He has been a member of Technical Program Committee of various conferences including IEEE INFOCOM, ICC, Globecom, ACM MobiHoc, and ACM Sigmetrics. He is currently on the editorial board of IEEE/ACM Transactions on Networking and Computer Communications Journal, and was TPC co-chair of WASA'11. He received the Best Paper Awards in the IEEE ICCCN 2005, IEEE IPCCC 2006, and IEEE NetSciCom 2015, and the National Science Foundation CAREER Award 2006. He supervised and co-authored a paper that received the Best Student Paper Award in ACM MobiCom 2007.  
\end{IEEEbiography}
\vspace{-10mm}
\begin{IEEEbiography}[{\includegraphics[width=1.15in,height=1.15in,clip]{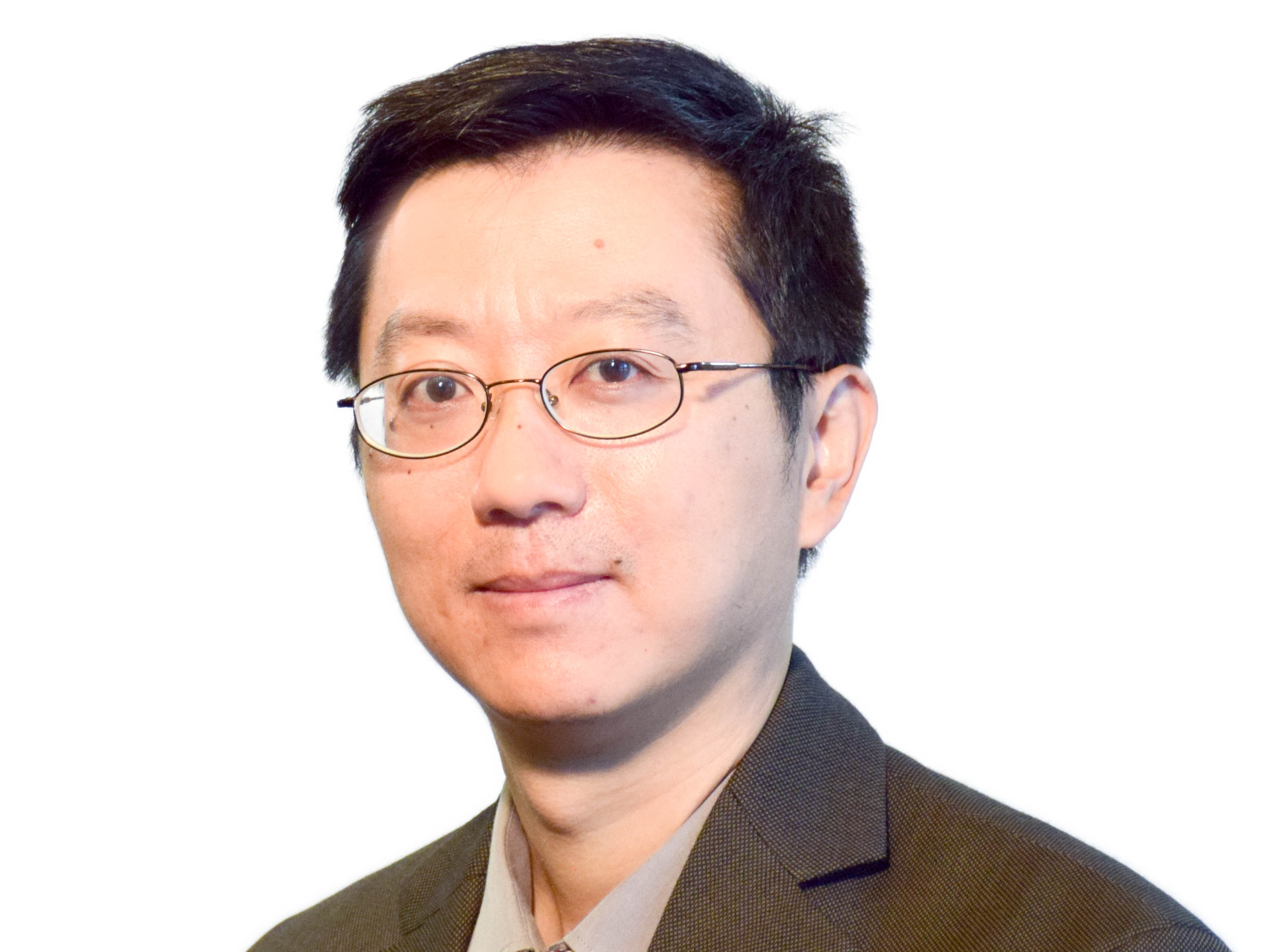}}]{Huaiyu Dai} (F’17)
  received the B.E. and M.S. degrees in electrical engineering from Tsinghua
University, Beijing, China, in 1996 and 1998, respectively, and the Ph.D. degree in electrical
engineering from Princeton University, Princeton, NJ in 2002. He was with Bell Labs, Lucent Technologies, Holmdel, NJ, in summer 2000, and with AT\&T Labs-Research, Middletown, NJ, in summer 2001. He is currently a Professor of Electrical and Computer
Engineering with NC State University, Raleigh, holding the title of University Faculty Scholar. His research interests are in the general areas of
communication systems and networks, advanced signal processing for digital communications, and
communication theory and information theory. His current research focuses on networked information
processing and crosslayer design in wireless networks, cognitive radio networks, network security, and
associated information-theoretic and computation-theoretic analysis.
He has served as an editor of IEEE Transactions on Communications, IEEE Transactions on Signal
Processing, and IEEE Transactions on Wireless Communications. Currently he is an Area Editor in
charge of wireless communications for IEEE Transactions on Communications. He co-edited two
special issues of EURASIP journals on distributed signal processing techniques for wireless sensor
networks, and on multiuser information theory and related applications, respectively. He co-chaired the
Signal Processing for Communications Symposium of IEEE Globecom 2013, the Communications
Theory Symposium of IEEE ICC 2014, and the Wireless Communications Symposium of IEEE
Globecom 2014. He was a co-recipient of best paper awards at 2010 IEEE International Conference on
Mobile Ad-hoc and Sensor Systems (MASS 2010), 2016 IEEE INFOCOM BIGSECURITY
Workshop, and 2017 IEEE International Conference on Communications (ICC 2017).
\end{IEEEbiography}
\end{document}